\documentclass[conference]{resources/IEEEtran}

\usepackage{marvosym}

\usepackage[T1]{fontenc}
\usepackage[utf8]{inputenc}
\usepackage[english]{babel}
\AtBeginDocument{%
	\providecommand\BibTeX{{%
			\normalfont B\kern-0.5em{\scshape i\kern-0.25em b}\kern-0.8em\TeX}}}

\usepackage{hyperref}
\usepackage{cite}

\usepackage{amsmath,amssymb,amsfonts} 
\usepackage[noend]{algpseudocode}
\usepackage{algorithmicx, algorithm}
\usepackage{graphicx} 
\usepackage{textcomp}
\usepackage{xcolor}

\usepackage{tikz}
\usetikzlibrary{arrows,decorations.markings, positioning,shapes}
\usetikzlibrary{trees} 
\usepackage{subcaption}
\usepackage{tabu}
\usepackage[normalem]{ulem}

\usepackage{resources/mathcommands}

\usepackage{resources/scheduling}

\usepackage{amsthm}

\newtheorem{thm}{Theorem}

\newtheorem{lem}[thm]{Lemma}
\newtheorem{cor}[thm]{Corollary}

\theoremstyle{definition}
\newtheorem{defn}[thm]{Definition}

\theoremstyle{remark}

\newtheorem{exmpl}[thm]{Example}

\newcommand{\DMR}{\mathrm{DMR}}
\newcommand{\DM}{\mathrm{DM}}
\let\oldLightning\Lightning
\renewcommand{\Lightning}{\text{\oldLightning}}

\newcommand{\Mod}[1]{\,\mathrm{mod}\,#1}

\begin{document}
\title{Dawn of the Dead(line Misses): \\Impact of Job Dismiss on the Deadline Miss Rate}
	

	\author{
		\IEEEauthorblockN{Jian-Jia Chen, Mario Günzel, Peter Bella, Georg von der Brüggen}
		\IEEEauthorblockA{TU Dortmund University, Germany\\
                 Email:
                  $\left\{
                  \begin{tabular}{l}
                   \mbox{jian-jia.chen, mario.guenzel}
                    \\\mbox{peter.bella, georg.von-der-brueggen}
                  \end{tabular}
                \right\}$@tu-dortmund.de}
                \and
		\IEEEauthorblockN{Kuan-Hsun Chen}
		\IEEEauthorblockA{University of Twente, the Netherlands\\
                 Email: k.h.chen@utwente.nl
                 }
            }

\maketitle

\begin{abstract}
  Occasional deadline misses are acceptable for soft real-time
  systems. Quantifying probabilistic and deterministic characteristics
  of deadline misses is therefore essential to ensure that deadline
  misses indeed happen only occasionally.
  This is supported by
  recent research activities on probabilistic worst-case execution
  time, worst-case deadline failure probability, the maximum number of
  deadline misses, upper bounds on the deadline miss
  probability, and the deadline miss rate.

  This paper focuses on the deadline miss rate of a periodic soft
  real-time task in the long run.  Our model assumes that this soft
  real-time task has an arbitrary relative deadline and that a job can still be executed after a
  deadline-miss until a dismiss point. This model generalizes the
  existing models that either dismiss a job immediately after its
  deadline miss or never dismiss a job.
  We provide mathematical notation on the convergence
  of the deadline miss rate in the long run and essential
  properties to calculate the deadline miss rate. Specifically, we use
  a Markov chain to model the execution behavior of a periodic soft
  real-time task. We present the required ergodicity property to
  ensure that the deadline miss rate in the long run is described by a
  stationary distribution.  
  %
\end{abstract} 
	
\section{Introduction}
\label{sec:introduction}

In classical hard real-time systems, deadline misses have to be
prevented under all circumstances.  However, in many industrial use
cases, occasional deadline misses can be
tolerated~\cite{DBLP:conf/rtss/AkessonNNAD20}.  Hence, providing
guaranteed quantification of deadline misses is important in
practice to ensure that deadline misses are indeed only occasional.
For example, safety standards such as IEC-61508 \cite{IEC} (Functional
safety) and ISO-26262 \cite{ISO} (Road vehicles - Functional safety)
specify an upper bound on the failure probability which can be very
low but not necessarily $0$. The importance of probabilistic guarantees 
is shown by
recent research activities on probabilistic worst-case execution time (c.f. the survey by Davis and
Cucu-Grosjean~\cite{DBLP:journals/lites/DavisC19}),
real-time queuing theory~\cite{DBLP:conf/rtss/Lehoczky97},
the maximum number of deadline misses over a number of task activations (e.g.,~\cite{DBLP:journals/tecs/SunN17}),
bounded tardiness (e.g., \cite{DBLP:conf/ecrts/AhmedA21}), and
probabilistic schedulability analysis~\cite{DBLP:journals/lites/DavisC19a}. According to the survey by Davis and
Cucu-Grosjean~\cite{DBLP:journals/lites/DavisC19a}, two probabilistic guarantees
are primarily studied: 
\begin{itemize}
\item The \emph{worst-case deadline failure probability} (WCDFP) is
  the maximal probability (among all jobs) that a job misses its
  deadline. Suppose that $\DM(j)$ is the probability that the $j$-th
  job misses deadline. WCDFP is $\sup_{j \in \Nbb} \DM(j)$. This
  problem has been studied in a series of
  results~\cite{1181583,DBLP:conf/rtas/TiaDSSSWL95,DBLP:journals/rts/LopezDEG08,6728877,10.1145/3139258.3139276,DBLP:conf/rtss/ChenGBC22}. Several
  approaches, such as down-sampling~\cite{5562910, 6728877,
    markovic_et_al:LIPIcs.ECRTS.2021.16, 10.1145/2392987.2393001},
  concentration inequalities~\cite{7993392,
    vonderbrggen_et_al:LIPIcs:2018:8997, 8714908}, task-level
  convolution~\cite{vonderbrggen_et_al:LIPIcs:2018:8997,von2021efficiently},
  and Monte-Carlo simulation~\cite{9622372} have been developed to
  improve the efficiency of the derivations of WCDFP.  Stochastic
  Network Calculus \cite{DBLP:journals/comsur/FidlerR15} extends Network Calculus 
  with   stochastic arguments using min-plus algebra, targeting properties
  related to probabilistic measures of the queue size and response time.

\item The \emph{deadline miss rate} (denoted as deadline miss
  probability\footnote{We explain in
    Section~\ref{sec:DRM-definition-literature} why the term
    probability may be misleading.} in the survey by Davis and
  Cucu-Grosjean~\cite{DBLP:journals/lites/DavisC19a}) is the number of
  jobs missing their deadlines divided by the number of released jobs 
  (informally, $\lim_{N\to\infty}\frac{1}{N}\sum_{j=1}^{N} \DM(j)$
  and formally elaborated in Eq.~(\ref{eq:DMR_lim_our}) in
  Section~\ref{sec:problem-definition}). 
  The deadline miss rate has been studied
  in~\cite{DBLP:journals/lites/DavisC19a,DBLP:conf/rtss/AbeniB98,DBLP:conf/ecrts/AbeniB99,DBLP:conf/etfa/ManicaPA12,DBLP:journals/jss/AbeniMP12,DBLP:conf/ecrts/PalopoliFMA12,DBLP:journals/tpds/PalopoliFAF16,DBLP:conf/rtcsa/ChenBC18,DBLP:conf/rtas/FriasPAF17,DBLP:conf/rtcsa/FriebePN20,DBLP:conf/rtcsa/FriebeMPN21,Friebe-RTAS2023,DBLP:conf/ecrts/ManolacheEP01,DBLP:conf/iccad/ManolacheEP02,DBLP:journals/tecs/ManolacheEP04,DBLP:conf/rtas/ManolacheEP04}  
  under different settings.
\end{itemize}

In this work, we focus on analyzing the deadline miss rate.
Intuitively, the goal is to answer the question: \emph{What is the
  ratio of jobs missing their deadline in the long run?}  Although
this questions seems quite straight forward, answering this question
has been a challenge over the last years.



For soft real-time tasks, the deadline miss rate was studied by Abeni et
al.~\cite{DBLP:conf/rtss/AbeniB98,DBLP:conf/ecrts/AbeniB99} when they
introduced the concept of constant bandwidth servers (CBS) in 1999.
Specifically, they studied two scenarios: 1) variable execution times
according to some probability distribution with fixed inter-arrival
times, and 2) variable inter-arrival times according to some
probability distribution with fixed execution times. Further
extensions were made to cover more complex scenarios, such as
probabilistic execution times and probabilistic job inter-arrival
times~\cite{DBLP:journals/jss/AbeniMP12,DBLP:conf/ecrts/PalopoliFMA12,DBLP:conf/etfa/ManicaPA12,DBLP:journals/tpds/PalopoliFAF16},
probabilistic execution times described via a Hidden Markov Model, and
execution times modeled by a Markov chain with continuous Gaussian
distributions
\cite{Friebe-RTAS2023,DBLP:conf/rtcsa/FriebeMPN21,DBLP:conf/rtcsa/FriebePN20}. \emph{These
  results all focus on a soft real-time task served independently by a
  CBS for temporal isolation.}  Specifically, Manica et
al.~\cite{DBLP:conf/etfa/ManicaPA12} and Palopoli et
al~\cite{DBLP:conf/ecrts/PalopoliFMA12,DBLP:journals/tpds/PalopoliFAF16}
use the Quasi-Birth-Death Process to efficiently compute the deadline
miss rate under CBS. These results assume that the
period of the periodic soft real-time task is an integer multiple of
the period of the CBS. Furthermore, after a soft real-time job misses its deadline,
it is executed until it finishes.

Manolache et
al.~\cite{DBLP:conf/ecrts/ManolacheEP01,DBLP:conf/iccad/ManolacheEP02,DBLP:journals/tecs/ManolacheEP04,DBLP:conf/rtas/ManolacheEP04}
study the deadline miss rate for task graphs with probabilistic
execution times under \emph{non-preemptive} scheduling
algorithms. In one line of work, they assume that a job is dismissed (aborted) immediately
after it misses its
deadline~\cite{DBLP:conf/ecrts/ManolacheEP01,DBLP:conf/iccad/ManolacheEP02,DBLP:conf/rtas/ManolacheEP04}.
In another, they consider systems with a limit on the number of jobs of a
soft real-time task which can still be executed beyond their
deadlines~\cite{DBLP:journals/tecs/ManolacheEP04}.

For preemptive fixed-priority scheduling, the deadline miss rate on a 
uniprocessor was studied by Chen et
al.~\cite{DBLP:conf/rtcsa/ChenBC18}. Their analysis is based on the
critical instant theorem to derive the WCDFP~\cite{7993392,6728877},
which has been recently refuted by Chen et al.~\cite[Section
V.B]{DBLP:conf/rtss/ChenGBC22}.

In a nutshell, the deadline miss rate for real-time tasks has been
studied in two categories: 1) under CBS and the assumption that every
job is executed till it finishes even after a deadline miss, or 2)
under non-preemptive scheduling and the assumption that a job is
aborted immediately after its deadline miss.  Extensions to other
scheduling policies and when and whether a job should still be
executed after its deadline miss have not been well explored.

Intuitively, from the perspective of minimizing the deadline miss
rate, a job should be immediately aborted after its deadline miss to ensure 
that it has no impact on the subsequent jobs. The reason is that, although the
result of a soft real-time job may be useful if it finishes after its deadline,
executing a job after its deadline may result in further jobs missing
their deadlines. For example, when the relative deadline of a periodic
task is the same as its period, a job which has already missed its
deadline for more than $b$ periods implies that (at least) the
$b$ subsequent jobs of the periodic task also miss their deadlines. 
Therefore, the system designers may specify a time
point at which the remaining workload of a soft real-time job has to
be dismissed. \emph{To the best of our knowledge, this paper is the first
one dealing with such a feature in a deadline miss rate analysis.}

In this paper, we
analyze the deadline miss rate of a periodic soft real-time task $\tau$.
The importance of this setting is supported by an empirical
study by Akesson et al.~\cite{DBLP:journals/rts/AkessonNNAD22} that examines industrial real-time systems, reporting that 
 82\% of the investigated systems have periodic task
 activations and that 67\% of the timing constraints are soft.
 We consider a soft real-time task that is served by a greedy
 processing component (GPC), which greedily processes an unfinished
 job of the soft real-time task in the ready queue in a
 first-come-first-serve manner whenever the GPC has capacity to
 process.  The information about the service offered by the GPC to
 execute the soft real-time task is modeled by supply functions.
 
\noindent\textbf{Our Contributions:}
\begin{itemize}
\item 
  In
  Section~\ref{sec:system_model:supply}, we
  adopt two types of supply functions and provide concrete examples
  on their applicability: 1)~\uline{\emph{Deterministic supply
      functions}}, e.g., under Time Division Multiple Access (TDMA) or
  preemptive fixed-priority scheduling with static execution times of
  the higher-priority tasks. 2) \uline{\emph{Supply bound functions}} with the upper and
  lower supply curves, 
  e.g., reservation servers. We note that this generalizes the studies
  in the literature, which focused on CBS or non-preemptive executions.
%
\item In Section~\ref{sec:problem-definition}, we introduce a
  mathematically rigorous definition of the deadline miss rate (DMR)
  in the long run. Further, we show in
  Section~\ref{sec:DRM-definition-literature} that the
  state-of-the-art definition given in the survey by Davis and
  Cucu-Grosjean~\cite{DBLP:journals/lites/DavisC19a} is a
  special case of ours.
\item In Section~\ref{sec:markov} we demonstrate how to represent deadline miss behavior through Markov chains. We present the theory behind the limiting behavior of Markov chains and quantify the deadline miss rate based on the limiting behavior in Section~\ref{sec:convergence-ergodicity}.
\item 
We provide
algorithms to construct Markov chains for GPCs specified with
deterministic and concrete supply functions in
Section~\ref{sec:DMR-supply-function} and upper/lower supply bound
functions in Section~\ref{sec:DMR-supplybound}.
Further, we provide examples to show how our results can be utilized. 
\item The methodology presented in the paper can be easily extended to different scenarios that are left out to simplify the presentation.
For example, assumptions on supply functions can be further weakened, and probabilistic supply functions can be treated similarly if a larger Markov chain can be affordable.
We discuss such extensions and the scalability of our approach in Section~\ref{sec:remarks}.
\end{itemize}

\section{System Model}
\label{sec:system_model}

Real-time systems can be \emph{hard} or \emph{soft}. A hard real-time
system does not tolerate any deadline miss, as it may result in
catastrophic consequences. A soft real-time system can tolerate
occasional deadline misses. Classically, depending on whether the
completion of a tardy job (after its deadline) has any utility to the
system, it is classified as either a firm real-time system (if the
tardy job has no utility at all) or a soft real-time system (if the
completion of the tardy job has certain utility).

In this paper, we focus on one
real-time task $\tau$, which is serviced by a greedy processing
component (see Section~\ref{sec:system_model:supply}) in a
uniprocessor system.  The task $\tau$ is modeled by a tuple
$(C, D, \delta, T)$, where $T > 0$ is the period
of $\tau$, $D > 0$ is its relative deadline, and
$\delta \geq 0$ is its \emph{relative dismiss point after deadline misses}. 
It 
releases an infinite number of successive task
instances, called jobs. The  \mbox{$j$-th} job of $\tau$ is denoted by $J_j$. 
We assume that the first job 
is released at time $0$. Hence, job $J_j$ is
released at time $(j-1)T$ and its absolute deadline is $(j-1)T+D$. If job
$J_j$ misses its deadline at $(j-1)T+D$, \emph{it is allowed to be further
executed} according to the relative dismiss point $\delta$ (that is, \uline{up to its absolute dismiss point} \mbox{$(j-1)T+D+\delta$}), and \uline{is dismissed afterwards}. We do not
assume any relationship of $D$ and $T$ (i.e., $\tau$ is an arbitrary-deadline task).

We note that this model generalizes the existing models that either
dismiss a job immediately after its deadline miss or never dismiss a
job. That is, with the configuration of the dismiss point after
deadline misses, we implicitly consider both the firm and soft
real-time task model in a general setting:
\begin{itemize}
\item For a firm real-time task, if the deadline is missed, then the
  job should be dismissed immediately, i.e., $\delta=0$.
\item For a soft real-time task, there are multiple  scenarios:
  \begin{itemize}
  \item If a job of the task is only considered useful up to a certain
    point after its deadline miss, then $\delta$ can be set to that
    specific point. We focus on the analysis of the deadline
    miss rate, provided that $\delta$ is specified.
  \item If a job cannot be dismissed until it finishes and the task's
    worst-case response time 
    is bounded, then $\delta$ is
    set to its worst-case response time minus its relative deadline.
  \item If a job cannot be dismissed until it finishes and the
    worst-case response time of the task is unbounded, then $\delta$
    is mathematically $\infty$. From the modeling perspective, this is
    also a feasible option, but our proposed method cannot be applied
    for this scenario as it may lead to an infinite Markov chain and the
    fundamental properties of this paper are based on finite Markov
    chains. 
    Thus, this case requires further
    explorations. However, practically, this may not be an interesting
    problem, as the scenario implies 
    any number of consecutive deadline misses can occur with
    a non-zero probability, resulting in a large interval of time where
    almost all jobs of this soft real-time task miss their
    deadlines. 
  \end{itemize}
\end{itemize}
For the rest of this paper, we call both of them soft real-time tasks,
as they both can tolerate occasional deadline misses.

$C$ is a random variable to describe the execution time of~$\tau$. 
 We assume that $C$ follows a
discrete distribution with \mbox{$h\geq 1$} distinct values
$e_1< e_2< \ldots< e_h$.  The execution times of the jobs
$J_j,\,j \in \Nbb\!=\!\setof{1,2,3,\dots}$ are described by the random variables
$C_j,\,j \in \Nbb$ which are independent copies of $C$. We denote by
$\mathbb{P}(X = x)$ the probability that a random variable $X$ is
equal to $x$. We assume that the actual execution time of job $J_j$ is
one of the given $h$ distinct values and that the sum of their
probabilities is $100\%$, i.e.,
$\sum_{k=1}^{h} \mathbb{P}(C_j = e_k) = 100\%$.  Furthermore, we
assume that the execution times of the soft real-time jobs are
\emph{independent and identically distributed} (iid).
Thus, their joint probability is equal to the product of their
probabilities. This is a commonly adopted  assumption  in the
literature, c.f.,~\cite{7993392,
  vonderbrggen_et_al:LIPIcs:2018:8997, 8714908, 6728877,
  DBLP:journals/lites/DavisC19a}.

\section{Supply Bound Functions}
\label{sec:system_model:supply}

The execution of the jobs of the soft real-time task $\tau$ is
abstractly modeled by a greedy processing component (GPC), which 
processes an unfinished job of $\tau$ 
in the ready queue in a first-come-first-serve (FCFS) manner whenever the GPC
has capacity to process.
The abstraction of GPC has
  been adopted in Real-Time Calculus~\cite{TCN00,WTVL06}. This allows
  us to model a variety of scheduling policies (see Examples~\ref{example:supply-function},
  \ref{example:supply-function-v2},
  \ref{example:supply-bound-functions}, and 
  \ref{example:supply-bound-functions-v2}).

To describe the capacity of the service
provided by the GPC to serve jobs of $\tau$, let $\beta_j(t)$
be the amount of accumulative service (supply) the GPC provides in the
time interval \mbox{$[(j-1)T, (j-1)T+t)$} for $0 \leq t \leq T$ and
\mbox{$j \in \Nbb\!=\!\setof{1,2,3,\dots}$}. By definition, $\beta_j(0)=0$,
$\beta_j(t)$ is non-decreasing, and
$\beta_j(t+\epsilon) \leq \beta_j(t) + \epsilon$ for any $\epsilon > 0 $
with $t+\epsilon \leq T$.
  
Depending on the accuracy of specification\footnote{For discussions on specific over general modelling, see~\cite{DBLP:conf/rtcsa/BruggenBCDR22}.}, 
we consider two different models of accumulative service for the
GPC. 

\begin{itemize}
\item \textbf{Supply Functions}: $\beta_j(t)$ is deterministic.
\item \textbf{Supply Bound Functions}: Let $\beta_j^u(t)$ (respectively,
  $\beta_j^l(t)$) be the upper (respectively, lower) supply curve
  such that
  \begin{equation}
    \label{eq:supply-bound-function}
    \beta_j^u(t) \leq \beta_j(t) \leq \beta_j^l(t)
  \end{equation}
  for any $0 \leq t \leq T$.
\end{itemize}

We now show some examples of these models to explain how they can be applied to certain scenarios:
\begin{exmpl}[\textbf{Supply Function}]
  \label{example:supply-function}
  Suppose the period $T$ of the soft real-time task $\tau$ is $3$,
  served by time division multiple access (TDMA), in which the TDMA
  cycle is $3$ and the service is from time $1$ to time $3$ (i.e., $2$
  units of time) within the TDMA cycle. In this case,
  $\forall 0 \leq t \leq 3$ and
  $\forall j \in \Nbb\!=\!\setof{1,2,3,\dots}$,
  \begin{align}
    \beta_j(t) & = \left\{
                  \begin{array}{ll}
                    0 & \mbox{ if }0 \leq t < 1\\
                    t-1 & \mbox{ otherwise.}\\
                  \end{array}
                  \right.
  \end{align}
  \qed
\end{exmpl}

\begin{exmpl}[\textbf{More Complex Supply Function}]
  \label{example:supply-function-v2}
  Suppose that there are two periodic tasks scheduled under preemptive
  fixed-priority scheduling in a uniprocessor system. The
  higher-priority task is a hard real-time task, releasing its first
  job at time $0$, with period $3$ and an actual execution time
  of $1$.
  Suppose the period $T$ of the lower-priority soft real-time task
  $\tau$ is $4$.  In this case, $\forall 0 \leq t \leq 4$ and
  $\forall j \in\setof{1,4,7,\dots} \subseteq \Nbb$,
  \begin{align}
    \beta_j(t) & = \left\{
                  \begin{array}{ll}
                    0 & \mbox{ if }0 \leq t < 1\\
                    t-1 & \mbox{ if } 1 \leq t < 3\\
                    2 & \mbox{ otherwise,}\\
                  \end{array}
                  \right.
  \end{align}
  $\forall 0 \leq t \leq 4$ and
  $\forall j \in \setof{2,5,8,\dots} \subseteq  \Nbb$,
  \begin{align}
    \beta_j(t) & = \left\{
                  \begin{array}{ll}
                    t & \mbox{ if }0 \leq t < 2\\
                    2 & \mbox{ if } 2 \leq t < 3\\
                    t-1 & \mbox{ otherwise,}\\
                  \end{array}
                  \right.
  \end{align}
  and
  $\forall 0 \leq t \leq 4$ and
  $\forall j \in\setof{3,6,9,\dots}\subseteq \Nbb$,
  \begin{align}
    \beta_j(t) & = \left\{
                  \begin{array}{ll}
                    t & \mbox{ if }0 \leq t < 1\\
                    1 & \mbox{ if } 1 \leq t < 2\\
                    t-1 & \mbox{ otherwise.}\\
                  \end{array}
                  \right.
  \end{align}  
  \qed
\end{exmpl}

\begin{exmpl}[\textbf{Supply Bound Functions}]
  \label{example:supply-bound-functions}
  Suppose the period~$T$ of the soft real-time task is $4$, served by
  a \emph{hard constant bandwidth server} (CBS), in which the CBS has
  a budget of $0.5$ and a period of $1$. We further assume that the
  CBS is guaranteed to provide the service.  In this case,
  $\forall 0 \leq t \leq 4$ and
  $\forall j \in \Nbb\!=\!\setof{1,2,3,\dots}$, the upper and lower
  supply curves are as follows:
  \begin{align}
    \beta_j^u(t) & = 0.5 \floor{t} + \min\{t-\floor{t}, 0.5\}.\\
    \beta_j^l(t) & = 0.5 \floor{t} + \max\{t-\floor{t}-0.5, 0\}.
  \end{align}
  \qed
\end{exmpl}

\begin{exmpl}[\textbf{More Complex Supply Bound Functions}]
  \label{example:supply-bound-functions-v2}
  The upper and lower supply bound functions of
  Example~\ref{example:supply-function-v2} are:
  \begin{itemize}
  \item $\forall j \in \setof{1,4,7,\dots} \subseteq \Nbb$
    \begin{itemize}
    \item $\beta_j^u(t) = t$ for $0 \leq t < 2$, and \\ $\beta_j^u(t) = 2$  for \mbox{$2 \leq t \leq 4$,}
    \item $\beta_j^l(t) = 0$ for $0 \leq t < 2$, and \\ $\beta_j^l(t) = t-2$
      for $2 \leq t \leq 4$,
    \end{itemize}
  \item $\forall j \in \setof{2,5,8,\dots} \subseteq \Nbb$ and
    $\forall j \in \setof{3,6,9,\dots} \subseteq \Nbb$,
    \begin{itemize}
    \item $\beta_j^u(t) = t$ for $0 \leq t < 3$, and \\ $\beta_j^u(t) = 3$
      for \mbox{$3 \leq t \leq 4$,}
    \item $\beta_j^l(t) = 0$ for $0 \leq t < 1$, and \\ $\beta_j^l(t) = t-1$
      for $l \leq t \leq 4$,
    \end{itemize}
  \end{itemize}
  \qed
\end{exmpl}


In the above examples, 
the given supply functions and 
the bounded supply
functions 
are repeated patterns
for every $Q \in \mathbb{N}$ jobs of the soft real-time tasks. 
Specifically, $Q$ is $3$ in
Examples~\ref{example:supply-function-v2}~and~\ref{example:supply-bound-functions-v2}
and $Q$ is $1$ in Examples~\ref{eq:supply-bound-function}~and~\ref{example:supply-bound-functions}. Such a repetitive pattern is necessary to determine the supply for an infinite sequence of jobs. For the rest of this paper, we assume that the supply functions and $Q$ are
specified.

\textbf{Remarks:} Modeling the service provided to a task using supply
(bound) functions has been widely studied in the literature 
(c.f. hierarchical scheduling in uniprocessor
systems~\cite{DBLP:conf/rtss/DavisB05,DBLP:conf/ecrts/SaewongRLK02}
and in multiprocessor systems~\cite{DBLP:conf/rtcsa/BiniBB09}, as well as
Network Calculus~\cite{LT01} and Stochastic Network
Calculus~\cite{DBLP:journals/comsur/FidlerR15}). However, to the best
of our knowledge, adopting supply (bound) functions for the purpose of
analyzing the deadline miss rate of a soft real-time task has never
been reported before. In the literature, there are only results based
on CBS or non-preemptive executions.

\emph{\uline{This paper sets its focus on the deadline miss rate
    analysis under the assumption that the supply functions or supply
    bound functions are specified.}} We have demonstrated several
useful cases and examples above, but how to derive the tightest supply
(bound) function is not the scope of this paper. However, deriving supply bound functions for
arbitrary interval lengths 
has been discussed in the literature,
e.g.,~\cite{WTVL06}. These methods can be altered to fit the needs
above.

\section{Problem Definition}
\label{sec:problem-definition}

This paper studies the following question:
\emph{Which percentage of deadline misses can be expected \uline{in the long run}?}
To answer this question, we first look at the definition for bounded intervals (Section~\ref{sec:problem-definition:bounded}) which we then extend to infinite intervals (Section~\ref{sec:problem-definition:long_run}). Afterwards, we discuss how our definition relates to definitions from the literature in Section~\ref{sec:DRM-definition-literature}.   

\subsection{Deadline Miss Rate of the First $N$ Jobs}\label{sec:problem-definition:bounded}
Let $\DM(j)$ be a random variable, indicating whether the \mbox{$j$-th} job $J_j$ of the soft real-time task $\tau$ misses its deadline. That is, $\DM(j)=1$ if $J_j$ has a deadline miss and $\DM(j)=0$ if $J_j$ successfully finishes until its deadline. Therefore, 
the \emph{deadline miss rate of the first}  $N \in \Nbb$ \emph{jobs} of the soft real-time task $\tau$ is a \emph{random variable} that is 
determined by counting the number of deadline misses and dividing by $N$: 
\begin{equation}
  \DMR_N = \frac{1}{N} \sum_{j=1}^{N} \DM(j)
\end{equation}

\begin{figure}
\vspace{-1.25cm}
  \centering
\tikzstyle{level 1}=[level distance=2cm, sibling distance=6cm]
\tikzstyle{level 2}=[level distance=2cm, sibling distance=3cm]
\tikzstyle{level 3}=[level distance=2.5cm, sibling distance=2cm]

\tikzstyle{end} = [circle, minimum width=3pt,fill, inner sep=0pt]

\begin{tikzpicture}[grow=right, sloped]
\scalebox{0.8}{
\node[end] { } 
    child {
        node[] {(\checkmark, 0)}        
            child {
                node[] {(\checkmark, 0)}
                child{
                  node[label=right:
                    {$\DMR_3=0$ with prob. $\frac{1}{8}$}] {(\checkmark, 0)}
                  edge from parent
                  node[above] {$C_3=3$}
                  node[below]  {$\frac{1}{2}$}
                }
                child{
                  node[label=right:
                    {$\DMR_3=0$ with prob. $\frac{1}{8}$}] {(\checkmark, 0)}
                  edge from parent
                  node[above] {$C_3=2$}
                  node[below]  {$\frac{1}{2}$}
                }
                edge from parent
                node[above] {$C_2=2$}
                node[below]  {$\frac{1}{2}$}
            }
            child {
                node[]
                {(\checkmark, 0)}
                child{
                  node[label=right:
                    {$\DMR_3=0$ with prob. $\frac{1}{8}$}] {(\checkmark, 0)}
                  edge from parent
                  node[above] {$C_3=3$}
                  node[below]  {$\frac{1}{2}$}
                }
                child{
                  node[label=right:
                    {$\DMR_3=0$ with prob. $\frac{1}{8}$}] {(\checkmark, 0)}
                  edge from parent
                  node[above] {$C_3=2$}
                  node[below]  {$\frac{1}{2}$}
                }                
                edge from parent
                node[above] {$C_2=3$}
                node[below]  {$\frac{1}{2}$}
            }
            edge from parent 
            node[above] {$C_1=2$}
            node[below]  {$\frac{1}{2}$}
    }
    child {
        node[] {(\Lightning, 1)}        
        child {
            node[]
            {(\checkmark, 0)}
            child {
              node[label=right:
              {$\DMR_3=\frac{1}{3}$ with prob. $\frac{1}{8}$}]
              {(\checkmark, 0)}
              edge from parent
              node[above] {$C_3=2$}
              node[below]  {$\frac{1}{2}$}
            }
            child {
              node[label=right:
              {$\DMR_3=\frac{1}{3}$ with prob. $\frac{1}{8}$}]
              {(\checkmark, 0)}
              edge from parent
              node[above] {$C_3=3$}
              node[below]  {$\frac{1}{2}$}
            }
                edge from parent
                node[above] {$C_2=2$}
                node[below]  {$\frac{1}{2}$}
          }
          child {
              node[]
              {(\Lightning, 1)}
              child {
                node[label=right:
                {$\DMR_3=\frac{2}{3}$ with prob. $\frac{1}{8}$}]
                {(\checkmark, 0)}
                edge from parent
                node[above] {$C_3=2$}
                node[below]  {$\frac{1}{2}$}
              }
              child {
                node[label=right:
                {$\DMR_3=1$ with prob. $\frac{1}{8}$}]
                {(\Lightning, 0)}
                edge from parent
                node[above] {$C_3=3$}
                node[below]  {$\frac{1}{2}$}
              }
              edge from parent
              node[above] {$C_2=3$}
              node[below]  {$\frac{1}{2}$}
          }
        edge from parent         
            node[above] {$C_1=3$}
            node[below]  {$\frac{1}{2}$}
          };}
        
      \end{tikzpicture}
\vspace{-1.25cm}
\caption{DMR of the first three jobs in
  Example~\ref{example:DMR-N-variable}.}
  \label{fig:DMR_after_3}
\end{figure}

\begin{exmpl}
  \label{example:DMR-N-variable}
  Consider a soft real-time task $\tau$ with $T=D=4$, $\delta=1$, and two
  possible of execution times
  \begin{itemize}
  \item $e_1=2$ and $e_2=3$, and
  \item $\Pbb(C_j = e_1) = 0.5$ and  $\Pbb(C_j = e_2) = 0.5$, $\forall j \in \Nbb$.
  \end{itemize}
  Suppose 
  $\tau$ is served by the supply function in
  Example~\ref{example:supply-function-v2}.
  
  Based on Figure~\ref{fig:DMR_after_3}, we can calculate the probability
  distribution of the random variable $\DMR_3$ by evaluating all
  scenarios. Each state
  $(\{\checkmark, \Lightning\},w)$ consists of two entries.  In
  particular, the first entry indicates if the job under consideration
  has a deadline miss (\checkmark{} for no deadline miss and
  \Lightning{} for deadline misses).
  The second entry $w$ is the amount of backlog (i.e., unfinished execution time)
  to be executed after the period and before its
  absolute dismiss point. That is, during $[4, 5)$ for $J_1$, during $[8, 9)$ for $J_2$,
  and during $[12, 13)$ for $J_3$.\footnote{This example could also be interpreted using $(\{\checkmark, \Lightning\})$ instead of $(\{\checkmark, \Lightning\},w)$ for each state. The reason why the backlog $w$ is introduced will be explained later in the paper.} If $C_1$ is $2$ for $J_1$, then no matter
  what happens with $J_2$ and $J_3$, the three jobs meet their
  deadlines. If $C_1$ is $3$ for $J_1$, then $J_1$ misses its
  deadline, and a backlog of one time unit must be executed in $[4, 5)$. Furthermore, if the execution time $C_2$ of $J_2$ is also
  $3$, then $J_2$ misses its deadline as well; otherwise if $C_2$ is
  $2$, then $J_2$ meets its deadline. Similarly, we can analyze the
  deadline miss of $J_3$ accordingly. We note that the scenario
  $C_1=3, C_2=3, C_3=3$ has one unit of backlog of $J_3$ at time $12$,
  but since the supply function does not provide any service to the
  soft real-time task $\tau$ from $12$ to $13$, the backlog that must be
  executed is set to $0$.

  We get the distribution $\Pbb(\DMR_3=0)=\frac{1}{2}$,
  $\Pbb(\DMR_3=\frac{1}{3})=\frac{1}{4}$, $\Pbb(\DMR_3=\frac{2}{3})=\frac{1}{8}$, and
  $\Pbb(\DMR_3=1)=\frac{1}{8}$.
  \qed
\end{exmpl}

\begin{figure*}
  \includegraphics[width=0.25\linewidth]{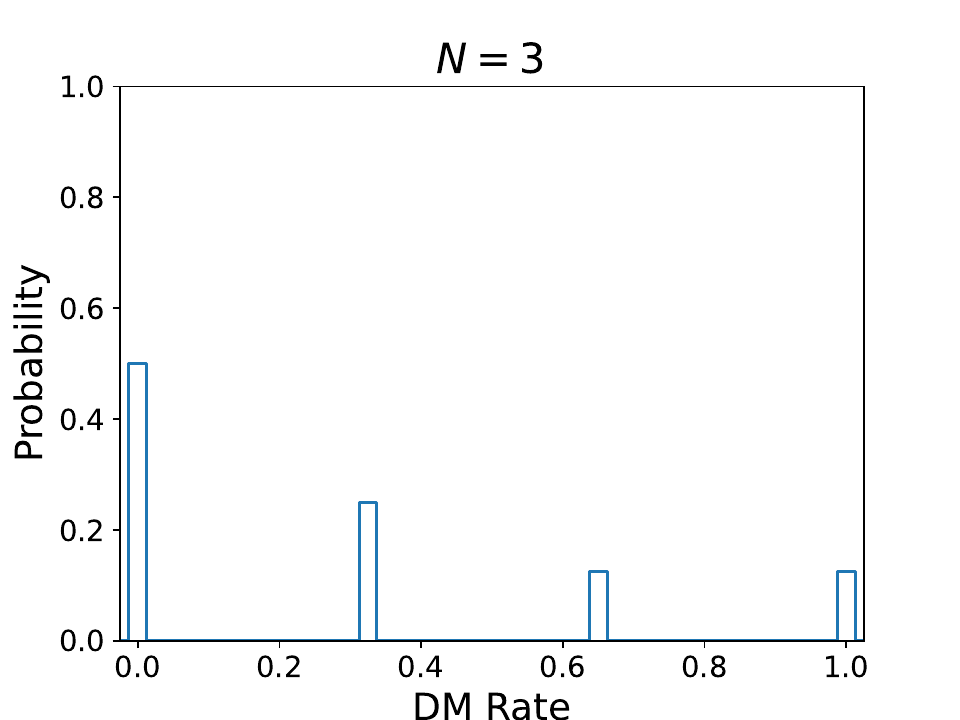}%
  \includegraphics[width=0.25\linewidth]{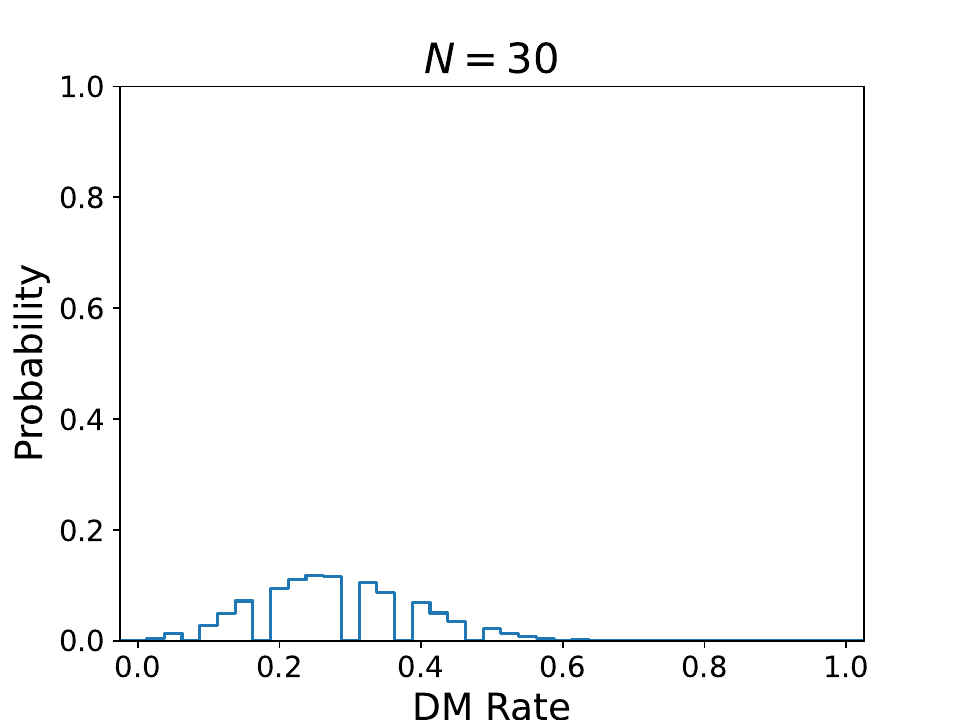}%
  \includegraphics[width=0.25\linewidth]{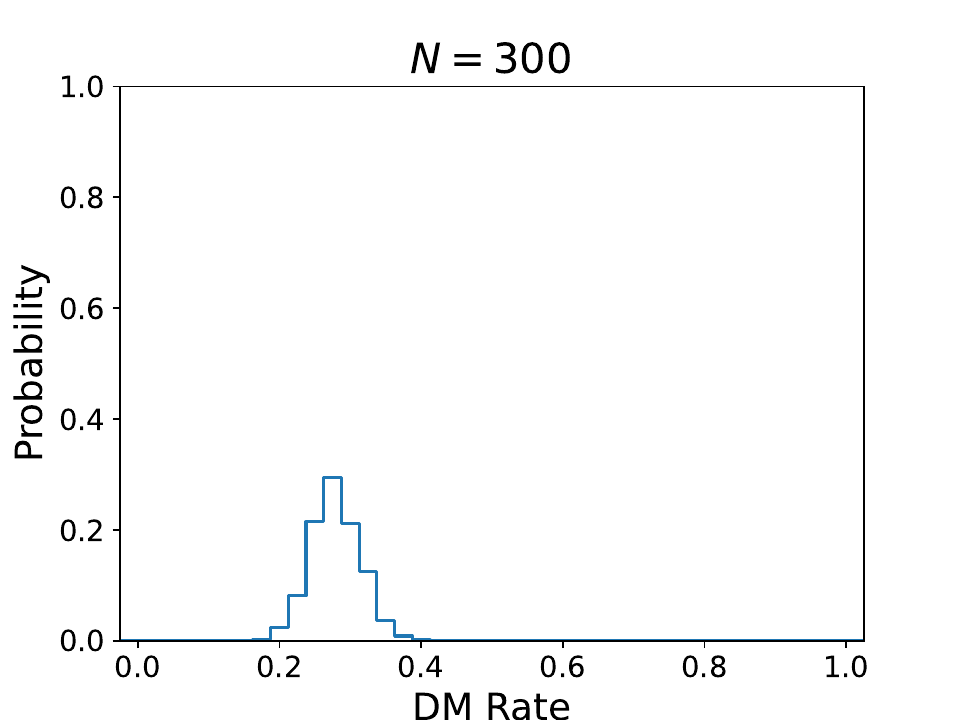}%
  \includegraphics[width=0.25\linewidth]{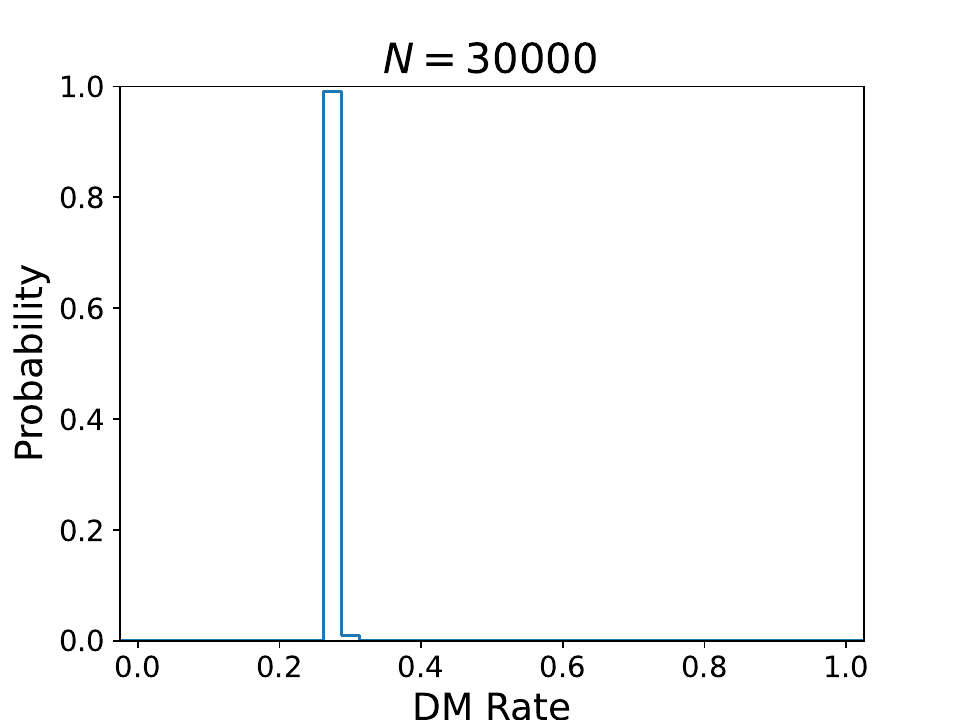}%
  \caption{Probability distribution of $\DMR_N$ of the first $N$ jobs. We observe that the distribution converges to a constant value.}
  \label{fig:DMR-N-approaching-convergence}
\end{figure*}

\subsection{Deadline Miss Rate in the Long Run}\label{sec:problem-definition:long_run}

Intuitively, the
deadline miss rate of the soft real-time task~$\tau$ in the long run is
simply the random variable $\DMR_N$ when $N\to \infty$. More
precisely, to derive the deadline miss rate in the long run, we need to
show that $\DMR_N$ converges towards a single value. 

Such convergence is a central property to derive the deadline miss rate.  
However, the existence and the computation of such a limiting
distribution is not trivial.  In this paper we answer the question:
\emph{Does $\DMR_N$ converge to a single value? If yes, to which value?}  If $\DMR_N$  almost surely converges to a constant value
\mbox{$\DMR\in \Rbb$}, then we say that $\DMR$ is the
\emph{deadline miss rate}.  Specifically, $\DMR$ satisfies the
following equation:
\begin{equation}\label{eq:DMR_lim_our}
  \Pbb\left( \lim_{N\to \infty} \DMR_N = \DMR\right) = 1
\end{equation}

In this paper we compute the deadline miss rate for the scenarios
covered in the system model.  Moreover, we determine conditions and
present the theory that justifies the existence of the deadline miss
rate. The following example demonstrates the convergence of
Example~\ref{example:DMR-N-variable}.

\begin{exmpl}
  \label{example:N-approches-infinity-converge}
  By extending Example~\ref{example:DMR-N-variable} for $N=30, 300, 30000$
  jobs, we observe that the probability distribution converges
  towards a single value as shown in
  Figure~\ref{fig:DMR-N-approaching-convergence}.
  \qed
\end{exmpl}

\subsection{DMR in the Literature}
\label{sec:DRM-definition-literature}

We utilize the recent survey by Davis and Cucu-Grosjean~\cite{DBLP:journals/lites/DavisC19a}, which defines the deadline miss rate (they call it the \emph{deadline miss probability}) for synchronous periodic tasks as follows:
\begin{equation}\label{eq:DMR_davis_cg}
  \DMR = \frac{1}{H} \sum_{j=1}^{H} \Pbb(\DM(j)=1)
\end{equation}
(restated from \cite[Definition~10]{DBLP:journals/lites/DavisC19a}), where $H$ is the number of jobs in one hyperperiod.
Their formula is actually a simplification of our definition and is a special case.
In their paper, all $H$ jobs are assumed to be served or dismissed until the $(H+1)$-th job is released. 
Therefore, under this assumption, the deadline miss rate of the $H$
jobs within a hyper-period is independent of the deadline miss rate of
the $H$ jobs in the previous hyper-period.
Moreover, the deadline miss rate of the first $H$ jobs and the $k$-th $H$ jobs, $k\geq 2$, are identically distributed. 
In the following, we show that the definition from Eq.~\eqref{eq:DMR_davis_cg} is a simplification of our definition from Eq.~\eqref{eq:DMR_lim_our}.

Let $Y_i = \frac{1}{H} \sum_{j=1}^H DM(j+(i-1)H)$ be the deadline miss rate of $H$ jobs in the $i$-th hyper-period of the soft real-time task $\tau$, i.e., the jobs from time $(i-1)HT$ to $iHT$.
Due to the repetitive pattern of periodic tasks, the random variables $Y_i, i=1,2, \dots$ are independent and identically distributed. 
Hence, the strong law of large number indicates that 
$\frac{1}{N}\sum_{i=1}^{N} Y_i$ converges almost surely to the expected value, i.e., 
\begin{equation}
  \Pbb\left(\lim_{N\to \infty} \frac{1}{N}\sum_{i=1}^{N} Y_i = \Ebb[Y_1]\right) = 1.
\end{equation}

If the limit of $\DMR_N = \frac{1}{N} \sum_{j=1}^{N} \DM(j)$ for $N\to \infty$ exists, then
\begin{align*}
  &\lim_{N\to \infty} \frac{1}{N} \sum_{j=1}^{N} \DM(j) 
  = \lim_{N\to \infty} \frac{1}{N \cdot H} \sum_{j=1}^{N\cdot H} \DM(j)\\
  = &\lim_{N\to \infty} \frac{1}{N} \sum_{i=1}^{N} \frac{1}{H} \sum_{j=1}^{H} \DM(j+(i-1)H)
   =  \lim_{N\to \infty} \frac{1}{N} \sum_{i=1}^{N} Y_i
\end{align*}
holds. 
Moreover, since the expected value is a linear and additive function even for dependent random variables,
\begin{align}
  \Ebb[Y_1] 
  = \frac{1}{H} \sum_{j=1}^H \Ebb[DM(j)]
  = \frac{1}{H} \sum_{j=1}^H \Pbb(DM(j)=1)
\end{align}
holds. 
We conclude that
\begin{align}
  \Pbb\left(\lim_{N\to \infty} \frac{1}{N}\sum_{j=1}^{N} \DM(j)  = \frac{1}{H} \sum_{j=1}^H \Pbb(DM(j)=1) \right) = 1
\end{align}
and the deadline miss rate is $\frac{1}{H} \sum_{j=1}^H \Pbb(DM(j)=1)$.

However, we note that Eq.~(\ref{eq:DMR_davis_cg}) cannot be applied
when not all $H$ tasks can be completely served or dismissed until the
$(H+1)$-th job is released.

We now also explain why we believe the term \emph{deadline miss rate}
is more appropriate than the term \emph{deadline miss probability}
used in the survey by
Davis~and~Cucu-Grosjean\cite{DBLP:journals/lites/DavisC19a}.  In
Eq.~(\ref{eq:DMR_lim_our}), although $\DMR_N$ can be interpreted as a
measure between $0\%$ and $100\%$, this is not a probability but
percentage of jobs with deadline misses when $N\to \infty$. Mathematically, in
Eq.~(\ref{eq:DMR_lim_our}), the event that we measure probabilistically happens almost
surely.\footnote{Almost surely means that the probability is $100\%$. }

\section{State Reduction and Markov Chains}
\label{sec:markov}

The soundness of our analysis is built on the well-established
foundation of Markov chains. To properly state the prerequisites of
the underlying properties, in this section, we provide the notation
used for Markov chains, following the work of
Norris~\cite[Chapter~1]{norris_1997}. Later, in
Section~\ref{sec:convergence-ergodicity}, we discuss the properties
for convergence and ergodicity.

We consider discrete time Markov chains denoted as
\mbox{$X_{\bullet} = (X_n)_{n\in \Nbb}$.}
Let $S$ be a countable set, the so-called \emph{state space}.
Each $X_n, n\in\Nbb$ is a random variable with values in $S$.
For our case, $X_n$ is the random variable that indicates the state of the $n$-th job. 
The probability that the $n$-th job is of state $s \in S$ is denoted by $\Pbb(X_n=s)$.

The transition from one state $X_n$ to the next state $X_{n+1}$ is described by a stochastic matrix\footnote{Specifically, we consider a left stochastic matrix, for which each entry $P_{r,s}$ is in $[0,1]$ and the columns add up to $1$.} $P = (P_{r,s})_{r,s \in S}$. 
In particular, for $r,s \in S$, if $X_n = s$ then the probability that $X_{n+1} = r$ is $P_{r,s} \in [0,1]$.

The initial distribution of the Markov chain is denoted as $\lambda =(\lambda_s)_{s\in S}$ and describes the probability distribution of $X_1$.
More specifically, $\Pbb(X_1 = s) = \lambda_s$ for all $s\in S$.

Fundamental for a Markov chain is the \emph{Markov property}.
That is, the probabilistic behavior of $X_{n+1}$ only depends on the result of $X_n$ and not of the preceding trace $X_1,\dots, X_{n-1}$.
This property is also called \emph{memoryless} for Markov chains.
More formally, the Markov property is fulfilled if
\begin{equation}\label{eq:markov_property}
  \begin{aligned}
    &\Pbb(X_{n+1} = s_{n+1} | X_1 = s_1, \dots, X_n=s_n) 
    \\&= \Pbb(X_{n+1} = s_{n+1} | X_n=s_n),
  \end{aligned}
\end{equation}
where $\Pbb(A | B) = \frac{\Pbb(A \cap B)}{\Pbb(B)}$ is the conditional probability.

Take the scenario depicted in Figure~\ref*{fig:DMR_after_3} as an
example. Each tuple of $(\{\mbox{\Lightning}, \checkmark\},w)_j$ is a
state of job $J_j$. Figure~\ref{fig:markov-complete-Example2} shows the
corresponding states and their transitions, in which the states in
$j$-th column represent for the possible states of $J_j$ for
$j=1,2,3$. Although Figure~\ref{fig:markov-complete-Example2} is only
for three jobs, it can be further extended to any positive integer
$n$. In this representation, there are two possible initial states in
the first column, each with a probability $50\%$, i.e., $\lambda$ is
specified. Furthermore, the Markov property is satisfied because the
state transition from $X_j$ to $X_{j+1}$ only depends on the
probability of $C_{j+1}$ and the state of $X_j$. Therefore, this is a
Markov chain with infinite states and can be represented by a matrix
with an infinite number of entries.

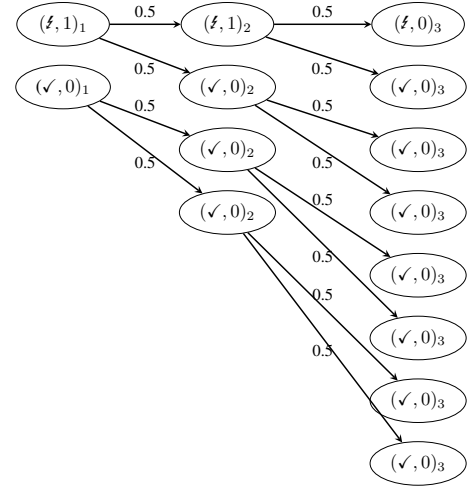
\begin{figure}[t]
  \centering
\scalebox{0.7}{

\begin{tikzpicture}[
  node distance = 0.35cm, auto, >=stealth,
  state/.style={ellipse, draw, minimum width=50pt},
  arrow/.style={->, thick}
]

\node[state] (node1A) {$(\Lightning, 1)_1$};
\node[state, below=of node1A] (node1B) {($\checkmark, 0)_1$};
\node[state, right=of node1A, xshift=1cm] (node2A) {$(\Lightning, 1)_2$};
\node[state, below=of node2A] (node2B) {$(\checkmark, 0)_2$};
\node[state, below=of node2B] (node2C) {$(\checkmark, 0)_2$};
\node[state, below=of node2C] (node2D) {$(\checkmark, 0)_2$};
\node[state, right=of node2A, xshift = 1.5cm] (node3A) {$(\Lightning, 0)_3$};
\node[state, below=of node3A] (node3B) {$(\checkmark, 0)_3$};
\node[state, below=of node3B] (node3C) {$(\checkmark, 0)_3$};
\node[state, below=of node3C] (node3D) {$(\checkmark, 0)_3$};
\node[state, below=of node3D] (node3E) {$(\checkmark, 0)_3$};
\node[state, below=of node3E] (node3F) {$(\checkmark, 0)_3$};
\node[state, below=of node3F] (node3G) {$(\checkmark, 0)_3$};
\node[state, below=of node3G] (node3H) {$(\checkmark, 0)_3$};

\node[right=of node3B, xshift=-1.8cm, yshift=-0.8cm] (node4B) {};

\draw[arrow] (node1A) to node[midway, above] {\small 0.5} (node2A);
\draw[arrow] (node1A) to node[midway, below] {\small 0.5} (node2B);

\draw[arrow] (node1B) to node[midway, above] {\small 0.5} (node2C);
\draw[arrow] (node1B) to node[midway, below] {\small 0.5} (node2D);

\draw[arrow] (node2A) to node[midway, above] {\small 0.5} (node3A);
\draw[arrow] (node2A) to node[midway, below] {\small 0.5} (node3B);

\draw[arrow] (node2B) to node[midway, above] {\small 0.5} (node3C);
\draw[arrow] (node2B) to node[midway, below] {\small 0.5} (node3D);

\draw[arrow] (node2C) to node[midway, above] {\small 0.5} (node3E);
\draw[arrow] (node2C) to node[midway, below] {\small 0.5} (node3F);

\draw[arrow] (node2D) to node[midway, above] {\small 0.5} (node3G);
\draw[arrow] (node2D) to node[midway, below] {\small 0.5} (node3H);
\end{tikzpicture}
}
  \caption{Markov chain of Figure~\ref{fig:DMR_after_3} for Example~\ref{example:DMR-N-variable}.}
  \label{fig:markov-complete-Example2}
\end{figure}

To efficiently represent the Markov chain, state reduction is
needed. One observation of the above example is that $J_2$ has only
two unique states $(\mbox{\Lightning}, 1)_2$ and $(\checkmark, 0)_2$ and $J_3$
has only two unique states $(\mbox{\Lightning}, 0)_3$ and $(\checkmark, 0)_3$.
Therefore, there are only two states for $J_1$, namely $s_{1,1}$ and
$s_{1,2}$, two states for $J_2$, namely $s_{2,1}$ and $s_{2,2}$, and
two states for $J_3$, namely $s_{3,1}$ and $s_{3,2}$. Furthermore,
since the two states of $J_3$ indicate that there is no workload of
$J_3$ being executed after time $12$, the execution behavior of $J_4$
is memoryless as there is no impact from
$J_3$. Figure~\ref{fig:markov-reduction-Example2} illustrates the
Markov chain (with
$S=\setof{s_{1,1}, s_{1,2}, s_{2,1}, s_{2,2}, s_{3,1}, s_{3,2}}$) of
Example~\ref{example:DMR-N-variable} for any arbitrary number of
iterations. The initial distribution of the Markov chain is
$\lambda_{s_{1,1}} = \lambda_{s_{1,2}}= 0.5$ and $\lambda_{s}= 0$ for
\mbox{$s \notin \setof{s_{1,1}, s_{1,2}}$.} The transition matrix $P$ is
\begin{equation}
  \label{eq:P-matrix-example}
P = \begin{pmatrix}
  0   & 0 & 0   & 0 & 0.5 & 0.5 \\
  0   & 0 & 0   & 0 & 0.5 & 0.5 \\
  0.5 & 0 & 0   & 0 & 0   & 0 \\
  0.5 & 1 & 0   & 0 & 0   & 0 \\
  0   & 0 & 0.5 & 0 & 0   & 0 \\
  0   & 0 & 0.5 & 1 & 0   & 0 
\end{pmatrix}
\end{equation}

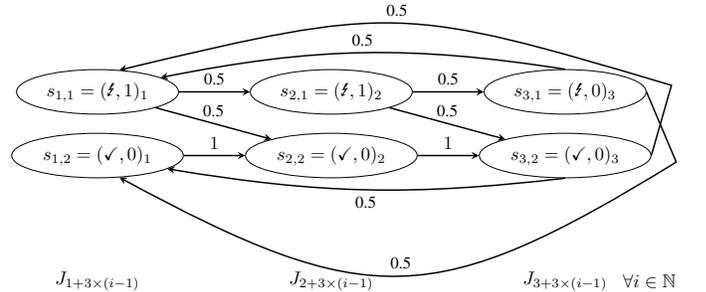
\begin{figure}[t]
  \centering
\scalebox{0.7}{

\begin{tikzpicture}[
  node distance = 0.35cm, auto, >=stealth,
  state/.style={ellipse, draw, minimum width=50pt},
  arrow/.style={->, thick}
]

\node[state] (node1A) {$s_{1,1} =(\Lightning, 1)_1$};
\node[state, below=of node1A] (node1B) {$s_{1,2}=(\checkmark, 0)_1$};
\node[state, right=of node1A, xshift=1cm] (node2A) {$s_{2,1}=(\Lightning, 1)_2$};
\node[state, below=of node2A] (node2B) {$s_{2,2}=(\checkmark, 0)_2$};
\node[state, right=of node2A, xshift=1cm] (node3A) {$s_{3,1}=(\Lightning, 0)_3$};
\node[state, below=of node3A] (node3B) {$s_{3,2}=(\checkmark, 0)_3$};

\node[below=of node1B, yshift=-1.3cm] (nodeJ1) {$J_{1+3\times (i-1)}$};
\node[below=of node2B, yshift=-1.3cm] (nodeJ2) {$J_{2+3\times (i-1)}$};
\node[below=of node3B, yshift=-1.3cm] (nodeJ3) {$J_{3+3\times (i-1)}$};
\node[right=of nodeJ3, xshift=-0.3cm] (nodeJ4) {$\forall i \in \Nbb$};

\node[right=of node3A] (node4A) {};
\node[right=of node3B] (node4B) {};

\draw[arrow] (node1A) to node[midway, above] {\small 0.5} (node2A);
\draw[arrow] (node1A) to node[midway, above] {\small 0.5} (node2B);

\draw[arrow] (node1B) to node[midway, above] {$1$} (node2B);

\draw[arrow] (node2A) to node[midway, above] {\small 0.5} (node3A);
\draw[arrow] (node2A) to node[midway, above] {\small 0.5} (node3B);

\draw[arrow] (node2B) to node[midway, above] {$1$} (node3B);

\draw[arrow] (node3A.north) to [bend right=10] node [midway, above]  {\small 0.5} (node1A);
\draw[arrow] (node3A.east) to (node4B.south) to [bend left=30, looseness=1.3] node[midway, above]  {\small 0.5} (node1B.315);

\draw[arrow] (node3B.east) to  (node4A.north) to [bend right=15, looseness=1.3] node [midway, above]  {\small 0.5} (node1A.45);
\draw[arrow] (node3B.south) to [bend left=8] node [midway, below]  {\small 0.5} (node1B);
\end{tikzpicture}
}
\caption{Finite Markov chain of an infinite amount of jobs for
  Example~\ref{example:DMR-N-variable}.}
  \label{fig:markov-reduction-Example2}
\end{figure}

\section{Convergence and Ergodicity}
\label{sec:convergence-ergodicity}

The convergence of the deadline miss rate can be traced back to a property called \emph{ergodicity} of the Markov chain.
Intuitively, ergodicity means that the ratio of visits of a state in the long run is described by the stationary distribution.
Ergodicity and its relation to limiting behavior of Markov chains has intensely been studied in the literature, e.g.,~\cite{douc2018markov,hernandez2012markov,norris_1997}.
In this work we mostly follow the notation of Norris~\cite[Chapter~1]{norris_1997}.

We now discuss the properties a Markov chain must satisfy to utilize ergodic theory.
Let $X_{\bullet} = (X_n)_{n\in \Nbb}$ be a Markov chain and let $P$ be the corresponding stochastic matrix.

\begin{defn}[Irreducible]
  Two states $r,s \in S$ \emph{communicate with each other} in $X_{\bullet}$ if they are reachable from one another with positive probability. More formally, $r$ and $s$ communicate with each other if there exist two sequences of states $r_1, \dots, r_\xi \in S$ and $s_1,\dots, s_\psi \in S$ such that:
  \begin{align}
    P_{r_1, r}, P_{r_2,r_1}, \dots, P_{r_{\xi}, r_{\xi-1}}, P_{s, r_{\xi}} &>0\\
    P_{s_1, s}, P_{s_2,s_1}, \dots, P_{s_{\psi}, s_{\psi-1}}, P_{r, s_{\psi}} &>0
  \end{align}

  $X_{\bullet}$ is called \emph{irreducible} if all states in $S$ communicate with each other in $X_{\bullet}$.
  \qed
\end{defn}
Intuitively, all nodes in the graph describing the Markov chain are connected by paths of non-zero probability.

\begin{defn}[Positive recurrent]
  We say that $X_{\bullet}$ is \emph{recurrent}, if for all states $s\in S$ 
  \begin{equation}
    \Pbb(X_n=s \text{ for infinitely many } n \in \Nbb\,|\,X_1=s) = 1
  \end{equation}
  holds. Intuitively, each state is infinitely many times visited.
  Let $T_s := \inf\set{n\geq 2}{X_n=s}$ be the first passage time\footnote{The first passage time is a random variable that describes how long it takes until the state $s$ is reached after the initial state. We assume in that definition that $\inf \emptyset = \infty$.} of state $s\in S$.
  The Markov chain $X_{\bullet}$ is \emph{positive recurrent} if $\Ebb[T_s | X_1 = s]<\infty$ for all $s\in S$.
  \qed
\end{defn}


If $X_{\bullet}$ is positive recurrent, then 
each state is expected to be visited again in finite time.
Equivalent descriptions of positive recurrent have been provided in the literature, e.g.,~\cite[Theorem~1.7.7]{norris_1997}.
%

If a Markov chain is irreducible and positive recurrent, then its ratio of visits in the long run is described by the \emph{stationary distribution}.

\begin{defn}[Stationary distribution]
  A probability distribution $\pi = (\pi_s)_{s\in S}$ with $\sum_{s\in S}\pi_s=1$ is \emph{stationary} if 
  \begin{equation}
    \label{eq:stationary-distribution}
    P\pi = \pi.
  \end{equation}
In the literature, the terms \emph{invariant} or \emph{equilibrium} are used
equivalently for the stationary distribution.
   \qed
\end{defn}

If the Markov chain is finite, a stationary distribution can be calculated by solving the linear system $(P-E)\pi = 0$ (where $E$ is the identity matrix with $1$ on the diagonal and $0$ else) and normalizing $\pi$ (that is, setting $\pi$ to $\frac{1}{\sum_{s\in S} \pi_s} \pi$).

The following theorem ensures the limiting behavior as well as the existence and uniqueness of the stationary distribution.

\begin{thm}[Ergodic Theorem. Reformulated from~{\cite[Theorem~1.10.2]{norris_1997}}]\label{thm:ergodic_norris}
  Consider a Markov chain $X_{\bullet} = (X_n)_{n\in \Nbb}$ with transition matrix $P$.
  Let $f: S \to \Rbb$ be a bounded function and $\lambda$ any initial distribution.
  If $X_{\bullet}$ is irreducible and positive recurrent, then a unique invariant distribution $\pi$ exists, and 
  \begin{equation}
    \Pbb\left(
      \lim_{N\to \infty} \frac{1}{N} \sum_{j=1}^{N} f(X_j) = \sum_{s\in S} \pi_s f(s)
    \right) = 1.
  \end{equation}
\end{thm}

By choosing the function $f$ to count the number of deadline misses, we can use the ergodic theorem to calculate the deadline miss rate.

\begin{thm}[Deadline miss rate]\label{thm:DMR}
  If $X_{\bullet}$ is irreducible and positive recurrent, then the unique invariant distribution $\pi$ exists, and 
  \begin{equation}
    \Pbb\left(
      \lim_{N\to \infty} \DMR_N = \sum_{s\in S_{\DM}} \pi_s
    \right) = 1,
  \end{equation}
  where $S_{\DM} \subseteq S$ are the states that indicate a deadline miss.
  Therefore, 
  \begin{math}
    \DMR = \sum_{s\in S_{\DM}} \pi_s.
  \end{math}
\end{thm}

\begin{proof}
  We choose the function $f : S\to \Rbb$ as $f(s) = 1$ if state $s$ indicates a deadline miss and $f(s) = 0$ if it indicates no deadline miss.
  In that case, $f(X_j)$ from Theorem~\ref{thm:ergodic_norris} is the same as $\DM(j)$, and $\sum_{s\in S} \pi_s f(s)$ can be simplified to $\sum_{s\in S_{\DM}} \pi_s$.
  By definition, $\DMR = \sum_{s\in S_{\DM}} \pi_s$.
\end{proof}

For our case, this means that if we can ensure that the Markov chain $X_{\bullet}$ is irreducible and positive recurrent, then we can calculate the deadline miss rate $\sum_{s\in S_{\DM}} \pi_s$ by calculating the stationary distribution $\pi$.

For finite Markov chains $X_{\bullet}$ (i.e., 
for chains with finitely many states $S$) 
whether $X_{\bullet}$ is irreducible can be determined efficiently. As all transitions are with non-zero probability, testing whether  $X_{\bullet}$ is irreducible is equivalent to testing whether the directed graph (Markov chain) is strongly connected, i.e., there is a path from every vertex to every other vertex. Tarjan's strong connected components algorithm~\cite{DBLP:journals/siamcomp/Tarjan72,DBLP:journals/ipl/NuutilaS94} solves this problem 
in linear time (with respect to the number of vertices and directed edges). If there is only one strongly connected component, then the Markov chain is irreducible.  

If the Markov chain is finite and irreducible then it is always positive recurrent and we do not need to check it by hand.
This leads to the following result for the finite case (cf.~\cite{wilmer2009markov}).

\begin{cor}[Deadline miss rate, finite case]\label{cor:DMR_finite}
  If $X_{\bullet}$ is finite and irreducible, then a unique invariant distribution $\pi$ exists.
  Moreover, it holds 
  \begin{equation}
    \Pbb\left(
      \lim_{N\to \infty} \DMR_N = \sum_{s\in S_{\DM}} \pi_s
    \right) = 1
  \end{equation}
  and therefore
  \begin{equation}
    \label{eq:DRM-pi}
    \DMR = \sum_{s\in S_{\DM}} \pi_s.
  \end{equation}
\end{cor}



The procedure to calculate the deadline miss rate when we have a finite Markov chain $X_{\bullet}$ is given by Algorithm~\ref{alg:procedure}.


\begin{algorithm}[t]
	\caption{Compute DMR.}
	 \footnotesize
	\label{alg:procedure}
	\begin{algorithmic}[1]
    \State \textbf{Input}: Stochastic matrix $P$ of finite Markov chain $X_{\bullet}$
    \State \textbf{Output:} $\DMR$ or $\mathit{None}$
    \State
    \State Check if $X_{\bullet}$ is irreducible.
    \If{$X_{\bullet}$ is not irreducible}
      \State \Return $\mathit{None}$
    \EndIf
    \State Calculate solution of $P\pi=\pi$ with $\sum_{s\in S} {\pi_s} = 1$.
    \State $\DMR := \sum_{s\in S_{\DM}} \pi_s$
    \State \Return $\DMR$
	\end{algorithmic}
\end{algorithm}

\begin{exmpl}
  \label{example:missrate-ergodic-stationary}
  Consider the Markov chain in
  Figure~\ref{fig:markov-reduction-Example2} for executing infinitely
  many jobs of $\tau$ in Example~\ref{example:DMR-N-variable}.
  This Markov chain is irreducible by observation since each node is reachable from every other node by a path of edges with positive probability.
  The set of states that indicate deadline misses is $S_{DM}= \setof{s_{1,1}, s_{2,1}, s_{3,1}}$. The
  stationary distribution $\pi$ such that $P\pi=\pi$ is
  $(\pi_{1,1}, \pi_{1,2}, \pi_{2,1}, \pi_{2,2}, \pi_{3,1}, \pi_{3,2}) = (4/24, 4/24, 2/24, 6/24, 1/24, 7/24)$. 
  With Corollary~\ref{cor:DMR_finite}, we obtain that the deadline miss rate is
  $\pi_{s_{1,1}} + \pi_{s_{2,1}} + \pi_{s_{3,1}}=7/24\approx 29.2\%$.  \qed
\end{exmpl}

\begin{exmpl}
  \label{example:missrate-ergodic-stationary-different-probability}
  We revise Example~\ref{example:DMR-N-variable} by setting 
  \begin{itemize}
  \item $\Pbb(C_j = 2) = p$ and $\Pbb(C_j = 3) = 1-p$,
    $\forall j \in \Nbb$
  \end{itemize}
  for some $0 < p < 1$.  We can construct the corresponding Markov
  chain, whose states are identical to
  Figure~\ref{fig:markov-reduction-Example2}, but with different
  initial state probabilities and different transition matrix $P$.
  We confirm that the finite Markov chain is irreducible.
  Hence, we can use the ergodic theory for all of them to calculate the DMR. We illustrate the corresponding DMR by
  Theorem~\ref{thm:DMR} in Figure~\ref{fig:missrate-different-prob-example} for $p=0.01$ to $p=0.99$ with step $0.01$. \qed
\end{exmpl}

\begin{figure}
  \centering
  \includegraphics[width=0.75\columnwidth]{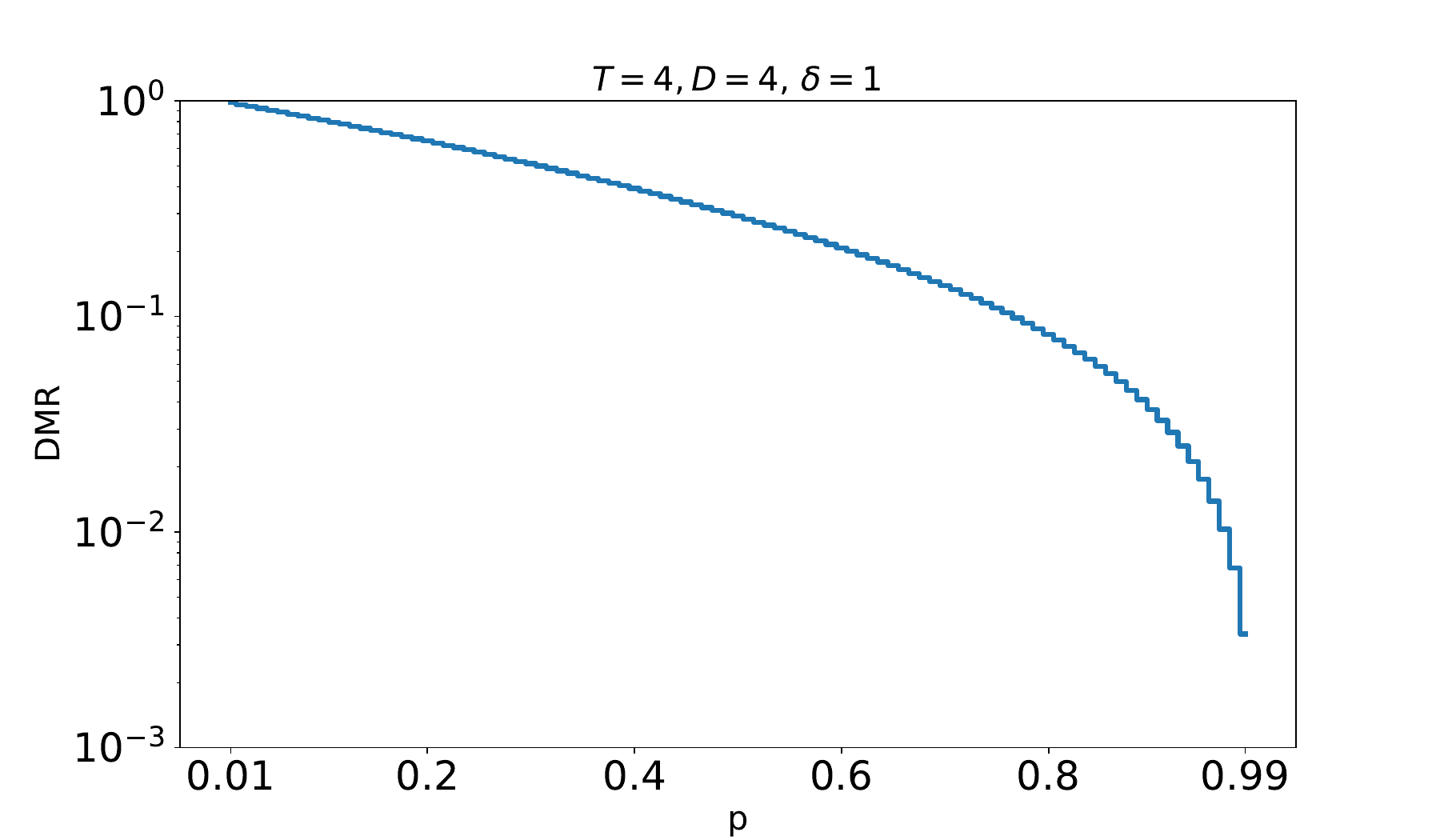}
  \caption{DMR of $\tau$ when $\Pbb(C_j = 2) = p$ and
    $\Pbb(C_j = 3) = 1-p$, $\forall j \in \Nbb$ in
    Example~\ref{example:missrate-ergodic-stationary-different-probability}.}
  \label{fig:missrate-different-prob-example}
\end{figure}


\section{DMR under Supply Functions}
\label{sec:DMR-supply-function}

We present the construction of the corresponding Markov chain in
Section~\ref{sec:algorithm-supply-functions} to capture the execution
of the task $\tau$ served by a GPC with a supply function $\beta_j(t)$
$\forall j \in \Nbb$ and analyze its correctness in
Section~\ref{sec:analysis-supply-functions}.  

The accumulative service
provided by the GPC starting from $(j-1)T$ to $(j-1)T+t$ is denoted as
$service(j, t)$:\footnote{We define$\sum_{\ell=j}^{j-1}  \beta_\ell(T)$ as $0$
  for notional brevity; otherwise, if $\floor{\frac{t}{T}}$ is $0$, 
  $\sum_{\ell=j}^{j+\floor{\frac{t}{T}}-1} \beta_\ell(T) =
  \sum_{\ell=j}^{j-1}  \beta_\ell(T)$ is mathematically undefined. }
\begin{equation}
  service(j, t) = \sum_{\ell=j}^{j+\floor{\frac{t}{T}}-1} \beta_\ell(T) +
  \beta_{j+\floor{\frac{t}{T}}}\left(t-T\floor{\frac{t}{T}}\right)
\end{equation}

\subsection{Algorithm to Construct $X_{\bullet}$}
\label{sec:algorithm-supply-functions}
\noindent\textbf{Initialization:} For each realization $e_k,~k=1,2,\ldots,h$ of
the random variable $C_1$ for the first job $J_1$ of $\tau$, one of
the following two cases holds:
\begin{itemize}
\item $J_1$ \uline{meets} its deadline for a realization $e_k$ of
  $C_1$ if \mbox{$e_k \leq service(1, D)$}. This realization $e_k$ of
  $C_1$ has
  \begin{equation}
    \label{eq:remaining-w1h}
  \max\{e_k - \beta_1(T), 0\}        
  \end{equation}
  remaining execution time that will be executed
  in the interval of $[T, \max\{D, T\})$.

\item $J_1$ \uline{misses} its deadline for a realization $e_k$ of
  $C_1$ if \mbox{$e_k > service(1, D)$}. Some of the remaining
  execution time $\max\{e_k - \beta_1(T), 0\}$ is dismissed after its
  dismiss point $D+\delta$ when $e_k > service(1, D+\delta)$ and $D+\delta > T$.
  Therefore,
  \begin{equation}
    \label{eq:remaining-w1h-dismissed}
   \max\{\min\{e_k,service(1, D+\delta)\} - \beta_1(T)\}, 0\}        
  \end{equation}
  remaining execution time will be executed in the interval of
  $[T, \max\{D+\delta, T\})$.
\end{itemize}

For any two realizations of $J_1$, if they have the same remaining
execution time 
then
they have the same impact on 
the subsequent jobs
$J_2, J_3, \ldots$.  Each state $(\{\checkmark, \Lightning\},w)_j$ of $J_j$
consists of two entries, indicating whether it is a deadline hit
$\checkmark$ and its remaining execution time is given by
Eq.~(\ref{eq:remaining-w1h}) or deadline miss $\Lightning$ and its
remaining execution time is given by Eq.~(\ref{eq:remaining-w1h-dismissed}).
If two states have the same deadline miss/hit indicator and the same remaining execution time they are merged to one state.
Among the $h$ realizations of $C_1$ for $J_1$, suppose that there are
$K_1 \leq h$ different states of $J_1$, denoted as
$s_{1,1}, s_{1,2}, \ldots, s_{1,K_1}$ with their corresponding initial
probabilities
$\lambda_{s_{1,1}}, \lambda_{s_{1,2}}, \ldots, \lambda_{s_{1,K_1}}$.

In the following, we assume that the supply functions repeat after every $Q \in \mathbb{N}$ jobs and 
 discuss the state expansion (i) for the first $Q$ jobs and (ii) for the subsequent jobs. 
This constructions of the states is achieved in a memoryless manner.

\noindent\textbf{State Expansions Up to $J_Q$:}
Let $j \in  \setof{1, 2, \ldots, Q-1}$.
For a given realization
$(e_{k_1}, e_{k_2}, \ldots, e_{k_j})$ of $C_1, C_2, \ldots, C_j$ of task~$\tau$, suppose that $s_{j, \ell}$
indicates this corresponding realization of $X_j$ and that
$rem(s_{j,\ell})$ is the remaining execution time of the state $s_{j,\ell}$
(assuming that execution time is dismissed after dismiss points). 
For each $e_k$ for $k=1,2,\ldots,h$ of the $h$
realizations of the random variable $C_{j+1}$ for the first job
$J_{j+1}$ of $\tau$, one of the following two cases holds:
\begin{itemize}
\item $J_{j+1}$ \uline{meets} its deadline for a realization $e_k$ of
  $C_{j+1}$ if \mbox{$e_k + rem(s_{j,\ell}) \leq service(j+1, D)$}. 
  The remaining execution time at $(j+1)T$ is
  \begin{equation}
    \label{eq:remaining-hit-j-supply-function}
    \max\{e_k+rem(s_{j,\ell}) - \beta_{j+1}(T), 0\}.
  \end{equation}
  This is similar
  to Eq.~(\ref{eq:remaining-w1h}) by taking $rem(s_{j,\ell})$ into
  considerations.
\item $J_1$ \uline{misses} its deadline for a realization $e_k$ of
  $C_{j+1}$ if \mbox{$e_k + rem(s_{j,\ell})> service(j+1, D)$}. 
  If further $rem(s_{j,\ell})+e_k > service(j+1, D+\delta)$ then 
  all execution time that is not served until $jT+D+\delta$ is dismissed.
  Thus, similar to Eq.~(\ref{eq:remaining-w1h-dismissed}), the
  remaining execution time is
  \begin{equation}
    \label{eq:remaining-miss-j-supply-function}    
    \max\left\{    
    \min\left\{
      \begin{array}{l}
        rem(s_{j,\ell})+e_k,\\
        service(j+1, D+\delta)   
      \end{array}
     \right\} - \beta_{j+1}(T),\, 0\right\}.            
  \end{equation}
\end{itemize}
Similarly, if two realizations of $C_{j+1}$ after $s_{j, \ell}$ are
identical, we can merge 
them into one. Suppose that there are
$K_{j+1}$ distinct states after the merge, denoted as
$s_{j+1, 1}, s_{j+1,2}, \ldots, s_{j+1, K_{j+1}}$.  The transition
probability from state $s_{j,\ell}$ to state $s_{j+1,\ell^*}$ is by
definition
$\sum_{\set{k\in \setof{1, \dots, h}}{s_{j,\ell} \text{ transits to } s_{j+1,
\ell^*} \text{ when }C_{j+1}=e_k}} \Pbb(C_{j+1}=e_k)$.

\noindent\textbf{State Expansions for $J_j = J_{Q+1}, J_{Q+2}, \ldots$:}
We now consider a given realization
$(e_{k_1}, e_{k_2}, \ldots, e_{k_j})$ for some $j \geq Q$. Since the
supply function repeats every $QT$ time units, there is no difference
of the execution behavior jobs $J_{j+1}$ and $J_{j\mod Q+1}$ if the
system has the same remaining execution time at time $jT$ and at time
$(j\mod Q)T$. 
Hence, \emph{what matters from the past} is only the
remaining execution time (after considering the dismiss points) and
\emph{what matters for the future} is the supply function, which 
repeats every $Q$ jobs. As a result, we can directly reuse an
existing state $s_{(j \mod Q)+1,\ell^*}$ if this is identical to a
state after executing $J_{j+1}$.

Starting from $j=Q$, if there is a state $s_{j,\ell}$ without any
out-going transition, we first evaluate the $h$ realizations of
$J_{j+1}$ with a similar procedure as in the state expansion up to $J_Q$. 
Suppose that there are $B$ distinct realizations, denoted as
$s_1', s_2', \ldots, s_B'$.
The following actions are taken to create states from the realizations $s_1', s_2', \ldots, s_B'$ and merge them into the Markov chain structure.
\begin{itemize}
\item If there is no state of $s_{(j \mod Q)+1, \ell^*}$ which is
  completely identical to $s_i'$, then a new state
  $s_{(j \mod Q)+1, \ell^*}$ is created and used to represent
  $s_i'$. The transition probability from $s_{j,\ell}$ to
  $s_{(j \mod Q)+1, \ell^*}$ is that from $s_{j,\ell}$ to $s_i'$.
\item If there is a state of $s_{(j \mod Q)+1, \ell^*}$ which is
  completely identical to $s_i'$, then $s_{(j \mod Q)+1, \ell^*}$ is
  used to represent~$s_i'$. 
  The transition probability from $s_{j,\ell}$ to
  $s_{(j \mod Q)+1, \ell^*}$ is that from $s_{j,\ell}$ to $s_i'$. We
  note that 
  the probability of outgoing transitions
  of $s_{(j \mod Q)+1, \ell^*}$
  remains as before since the
  transition of states is memoryless.
\end{itemize}

The above procedure repeats until there is no state $s_{j,\ell}$
without any out-going transition. Then, we further continue with $j$
as $(j \mod Q) +1$ until every state has out-going transitions.

\begin{exmpl}
  \label{example:DMR-example-D=6}
  We adopt Example~\ref{example:DMR-N-variable} with
  a small modification, 
  setting $D=6$ and $\delta=0$. The resulting Markov
  chain 
  is illustrated in
  Figure~\ref{fig:markov-reduction-ExampleD=6}. 
  The construction of
  $s_{1,1}, s_{1,2}, s_{2,1}, s_{2,2}, s_{3,1}, s_{3,2}$ is only
  slightly different from the six states in
  Example~\ref{example:missrate-ergodic-stationary}, and the only difference is that now 
  all jobs represented by states $s_{1,1}, s_{2,1}, s_{3,1}$ hit their deadline since they finish their workload before the deadline.

  We start the demonstration of our algorithm by continuing the construction from $j=Q=3$, i.e., to
  evaluate the execution behavior of $J_4$. Since $rem(s_{3,2})$ is
  $0$, state $s_{3,2}$ has no impact on the schedule of $J_4$,
  resulting in transiting into the two existing states $s_{1,1}$ and
  $s_{1,2}$ of $J_1$ each with 50\% probability. We move on with $s_{3,1}$.
  If $C_4$ is $2$ (50\%
  probability), since $rem(s_{3,1})$ is $1$, $J_4$ can be finished
  before its deadline and has remaining execution time of $2+1-2=1$
  at time $4T$. This results in an identical state as $s_{1,1}$ and
  the transition probability from $s_{3,1}$ to $s_{1,1}$ is $50\%$. If
  $C_4$ is $3$, since $rem(s_{3,1})$ is $1$, $J_4$ can be finished
  before its deadline and has remaining execution time of $3+1-2=2$
  at time $4T$. As this is different from $s_{1,1}$ and $s_{1,2}$, a
  new state $s_{1,3}$ is created and the transition probability from
  $s_{3,1}$ to $s_{1,3}$ is $50\%$.

  We now move further with $j=(3 \mod Q) +1=1$, i.e., to evaluate the
  execution behavior of $J_5$. As $s_{1,3}$ is the only state without
  out-going transitions, we consider the impact of the two
  realizations of $C_5$, following $s_{1,3}$. This results in a
  transition to an existing state $s_{2,1}$ with $50\%$ probability
  when $C_5 = 2$. When $C_5=3$, job $J_5$ misses its deadline as two
  units of its execution time must be after $5T$ but the service
  provided by the GPC between $5T$ and $4T+D=5T+2$ in this example is
  only one unit of time. Therefore, $J_5$ misses its deadline and the
  remaining one unit of execution time is dismissed at time $5T+2$
  since $\delta$ is $0$. This results in a new state
  $s_{2,3} = (\Lightning, 1)_2$ with $50\%$ probability when $C_5 = 3$.

  We now go ahead and consider $j=2$, i.e., to evaluate the execution
  behavior of $J_6$. As $s_{2,3}$ is the only state without out-going
  transitions, we consider the impact of the two realizations of
  $C_6$, following $s_{2,3}$. This results in a transition to an
  existing state $s_{3,1}$ with $50\%$ probability when $C_6 = 3$ and
  a transition to an existing state $s_{3,2}$ with $50\%$ probability
  when $C_6 = 2$.

  After that, since there is no state without out-going transitions,
  our construction terminates and the Markov chain $X_{\bullet}$ is
  returned. $X_{\bullet}$ is irreducible because all states are connected by a path with positive probability. \qed
\end{exmpl}

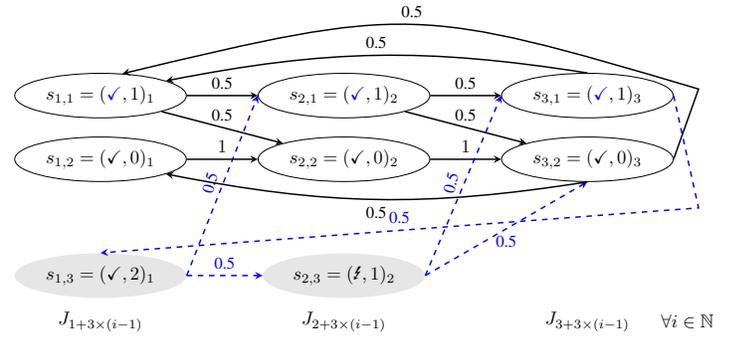
\begin{figure}[t]
  \centering
\scalebox{0.7}{

\begin{tikzpicture}[
  node distance = 0.35cm, auto, >=stealth,
  state/.style={ellipse, draw, minimum width=50pt},
  stateN/.style={ellipse, fill=gray!20, minimum width=50pt},
  arrow/.style={->, thick}
]

\node[state] (node1A) {$s_{1,1} =({\color{blue}\checkmark}, 1)_1$};
\node[state, below=of node1A] (node1B) {$s_{1,2}=(\checkmark, 0)_1$};
\node[stateN, below=of node1B, yshift=-1cm] (node1C) {$s_{1,3}=(\checkmark, 2)_1$};
\node[state, right=of node1A, xshift=1cm] (node2A) {$s_{2,1}=({\color{blue}\checkmark}, 1)_2$};
\node[state, below=of node2A] (node2B) {$s_{2,2}=(\checkmark, 0)_2$};
\node[stateN, below=of node2B, yshift=-1cm] (node2C) {$s_{2,3}=(\Lightning, 1)_2$};
\node[state, right=of node2A, xshift=1cm] (node3A) {$s_{3,1}=({\color{blue}\checkmark}, 1)_3$};
\node[state, below=of node3A] (node3B) {$s_{3,2}=(\checkmark, 0)_3$};

\node[below=of node1B, yshift=-2cm] (nodeJ1) {$J_{1+3\times (i-1)}$};
\node[below=of node2B, yshift=-2cm] (nodeJ2) {$J_{2+3\times (i-1)}$};
\node[below=of node3B, yshift=-2cm] (nodeJ3) {$J_{3+3\times (i-1)}$};
\node[right=of nodeJ3, xshift=-0cm] (nodeJ4) {$\forall i \in \Nbb$};

\node[right=of node3A] (node4A) {};
\node[right=of node3B, yshift=-0.8cm] (node4B) {};

\draw[arrow] (node1A) to node[midway, above] {\small 0.5} (node2A);
\draw[arrow] (node1A) to node[midway, above] {\small 0.5} (node2B);

\draw[arrow] (node1B) to node[midway, above] {$1$} (node2B);

\draw[arrow] (node2A) to node[midway, above] {\small 0.5} (node3A);
\draw[arrow] (node2A) to node[midway, above] {\small 0.5} (node3B);

\draw[arrow] (node2B) to node[midway, above] {$1$} (node3B);

\draw[arrow,->,dashed, blue] (node1C.east) to node[midway, above, rotate=75] {\small 0.5} (node2A.west);
\draw[arrow,->,dashed, blue] (node1C.east) to node[midway, above] {\small 0.5} (node2C.west);
\draw[arrow,->,dashed, blue] (node2C.east) to node[midway, above, rotate=75] {\small 0.5} (node3A.west);
\draw[arrow,->,dashed, blue] (node2C.east) to node[midway, below] {\small 0.5} (node3B.south);

\draw[arrow] (node3A.north) to [bend right=10] node [midway, above]  {\small 0.5} (node1A);
\draw[arrow,->, dashed, blue] (node3A.east) to (node4B.south) to node[midway, above]  {\small 0.5} (node1C.north);

\draw[arrow] (node3B.east) to  (node4A.north) to [bend right=15, looseness=1.3] node [midway, above]  {\small 0.5} (node1A.45);
\draw[arrow] (node3B.south) to [bend left=10] node [midway, below]  {\small 0.5} (node1B);

\end{tikzpicture}
}
\caption{Finite Markov chain of an infinite amount of jobs for
  Example~\ref{example:DMR-N-variable} by modifying with $D=6$ and
  $\delta=0$.}
  \label{fig:markov-reduction-ExampleD=6}
\end{figure}

\begin{exmpl}
  \label{example:DMR-example-varying-delta-And-D}
  We again adopt Example~\ref{example:DMR-N-variable} with two small
  modifications: 
  (a) $D=4$, $\delta = 0, 1,2,\ldots,10$ and
  (b) $\delta=1$ and $D=4,5,6,\ldots,14$. We create the finite Markov chains
  according to our algorithm and 
  confirm that they are all irreducible. Their DMRs are shown in
  Figure~\ref{fig:DMR-example-varying-delta-And-D}. 
  We observe that by enlarging the deadline $D$, the DMR decreases.
  However, when enlarging the relative dismiss point $\delta$, then DMR first increases but then seems to converge to a fixed value.
  \qed
\end{exmpl}

\begin{figure}[t]
  \centering
  \subfloat[DMR by varying $\delta$]{
    \includegraphics[width=0.45\linewidth]{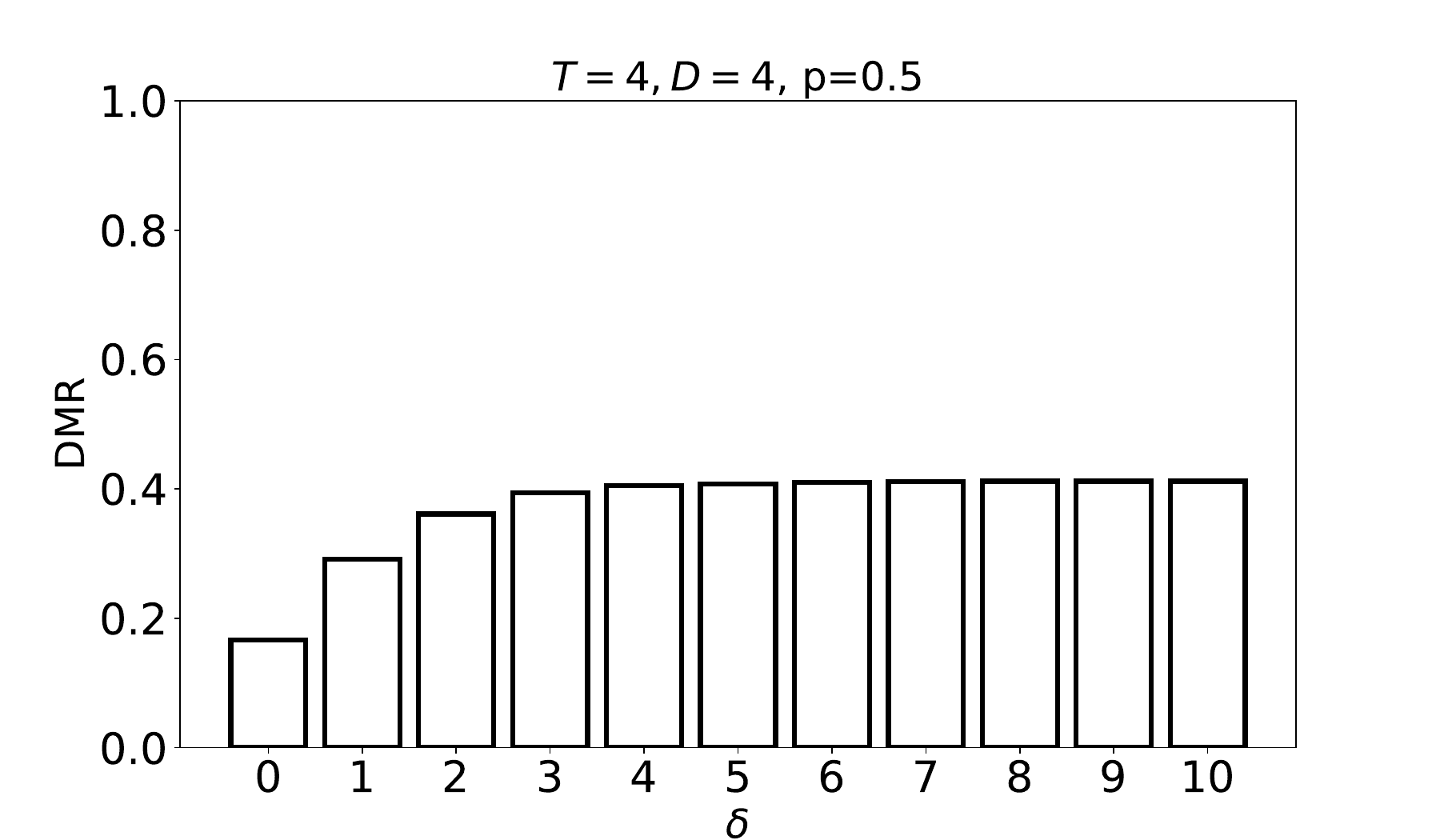}
  }%
  \subfloat[DMR by varying $D$]{
    \includegraphics[width=0.45\linewidth]{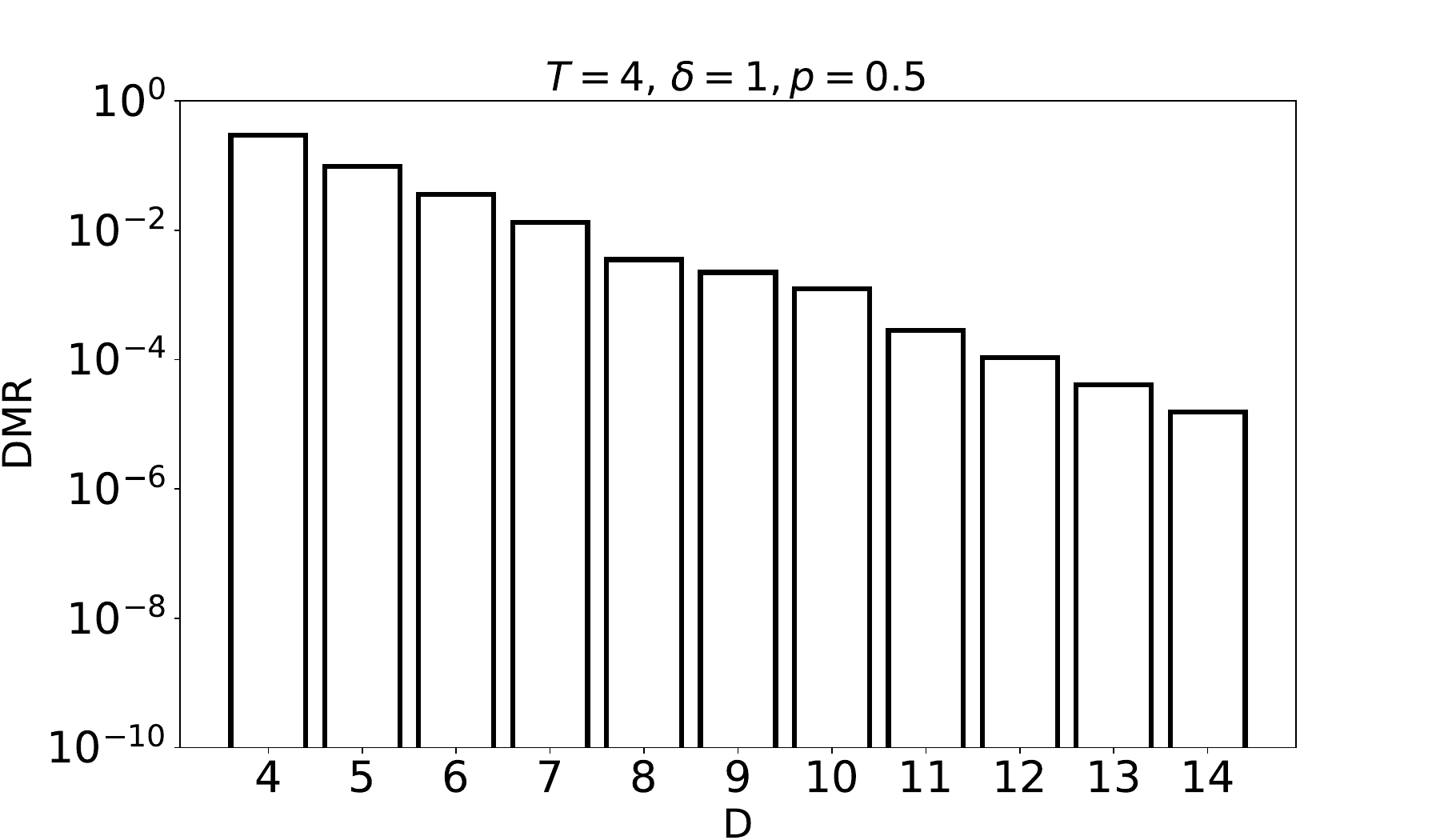}
  }
  \caption{DMR for Example~\ref{example:DMR-example-varying-delta-And-D}.}
  \label{fig:DMR-example-varying-delta-And-D}
\end{figure}

\subsection{Analysis of the Constructed Markov Chain}
\label{sec:analysis-supply-functions}

Let the resulting Markov chain be $X_{\bullet}$.  We now analyze the
time complexity of the construction process. Suppose that all variables are integers and the service
provision (if there is any) is also an integer. Let $W$ be
$\max_{j=1,2,\ldots,Q} service(j, D+\delta)$. The remaining execution
time of every state is therefore in $\setof{0, 1, 2, \ldots, W}$.
For every $j = 1, 2, \ldots, Q$, there are at most $2 (W+1)$ states of
$s_{j,*}$ in $X_{\bullet}$.  Finding whether a state exists and
merging afterwards can be implemented using the classic union-find
data structure~\cite{DBLP:journals/jacm/Tarjan75}, resulting in
$\Ocal(m\alpha(n))$ time complexity for $m$ operations on $n$ nodes, where
$\alpha(n)$ is the extremely slow-growing inverse Ackermann
function.

For every $j = 1, 2, \ldots$, every state $s_{j\mod Q, \ell}$ is evaluated with the $h$ realizations of the execution
times of $\tau$ after its
construction. For a given $j$, this results in at most $\Ocal(hW)$ find
operations and at most $\Ocal(W)$ state creations. If there is no new
state created for a given $j$, then the construction
finishes. Since there are most $\Ocal(W)$ states, the construction
terminates with $j=\Ocal(W)$. 
Therefore, provided
that $service(j+1, D+\delta)$ and $service(j+1, D)$ can be derived in
$\Ocal(1)$ time, the time complexity to construct $X_{\bullet}$ is
$\Ocal(hW^2\alpha(W))$.

Following the construction of the Markov chain, there is a canonical mapping from the realizations of $C_1, \dots, C_N$ to the realizations (states) of $X_1, \dots, X_N$.
This mapping preserves probabilities in the sense that
given a realization $(s_{1 \Mod{Q},\ell_1}, \dots, s_{N \Mod{Q},\ell_N})$ of $X_1, \dots, X_N$, 
the set of all realizations of $C_1, \dots, C_N$ that lead to $(s_{1\Mod{Q},\ell_1}, \dots, s_{N\Mod{Q},\ell_N})$ has the probability \mbox{$\Pbb(X_1=s_{1\Mod{Q},\ell_1}, \dots, X_N=s_{N\Mod{Q},\ell_N})$.}%
\footnote{The proof that this mapping is well-defined, surjective, and preserves the probability is left out to meet the space limitation. However, those properties hold mainly by the construction of the Markov chain.}
In the following, we show that this mapping does preserve the deadline miss rate, and therefore the constructed Markov chain can be used to calculate the deadline miss rate.

\begin{lem}
  \label{lemma:MarkovProperty-SupplyFunctions}
  Let $N\in \Nbb$, let $(s_{1\Mod{Q},\ell_1}, \dots, s_{N\Mod{Q},\ell_N})$ be a realization of $X_1, \dots, X_N$, and let $\Psi$ be the number of states that indicate deadline misses in $(s_{1\Mod{Q},\ell_1}, \dots, s_{N\Mod{Q},\ell_N})$.
  All realization of $C_1, \dots, C_N$ that lead to the realization $(s_{1\Mod{Q},\ell_1}, \dots, s_{N\Mod{Q},\ell_N})$ result in a job sequence of exactly $\Psi$ deadline misses.
\end{lem}

\begin{proof}
  We prove that this lemma holds for all $N\in \Nbb$ by induction over $N$.

  \textbf{Base state ($N=1$):} 
  Let $s_{1\Mod{Q},\ell_1}$ be a realization of $X_1$. 
  By construction, all realization of $C_1$ lead to jobs that have the same backlog and the same deadline miss/hit behavior as indicated by $s_{1\Mod{Q},\ell_1}$.


  \textbf{Induction step ($N-1 \mapsto N$):}
  Let $(s_{1\Mod{Q},\ell_1}, \dots, s_{N\Mod{Q},\ell_N})$ be a realization of $X_1, \dots, X_N$. 
  Then $(s_{1\Mod{Q},\ell_1}, \dots, s_{N-1\Mod{Q},\ell_{N-1}})$ is a realization of $X_1, \dots, X_{N-1}$.
  Let $\Rcal$ be the set of all realizations of $C_1, \dots, C_N$ that lead to $(s_{1\Mod{Q},\ell_1}, \dots, s_{N\Mod{Q},\ell_{N}})$.
  The first $N-1$ entries of any realization in $\Rcal$ lead to $(s_{1\Mod{Q},\ell_1}, \dots, s_{N-1\Mod{Q},\ell_{N-1}})$.
  Therefore, the first $N-1$ jobs of any job sequence obtained by $r \in \Rcal$ have the same number of deadline misses as $(s_{1\Mod{Q},\ell_1}, \dots, s_{N-1\Mod{Q},\ell_{N-1}})$ by induction.
  If the realizations of $C_N$ lead to a job sequence of an additional deadline miss, then 
  $s_{N\Mod{Q},\ell_{N}}$ indicates a deadline miss as well, by construction of the Markov chain.
  Similarly, if the realizations of $C_N$ lead to a job sequence of an additional deadline hit, then 
  $s_{N\Mod{Q},\ell_{N}}$ indicates a deadline hit as well.
  This proves that the number of deadline misses in $\Rcal$ are the same as the number of deadline misses in $(s_{1\Mod{Q},\ell_1}, \dots, s_{N\Mod{Q},\ell_{N}})$.
\end{proof}




\begin{thm}
  \label{theorem:DMR-SupplyFunctions}
  If $X_{\bullet}$ generated by using our algorithm for the given
  supply functions is irreducible then the deadline
  miss rate of the periodic soft real-time task $\tau$ is
  Eq.~(\ref{eq:DRM-pi}).
\end{thm}

\begin{proof}
  By Lemma~\ref{lemma:MarkovProperty-SupplyFunctions} the number of deadline misses in $X_{\bullet}$ coincide with the number of deadline misses of task $\tau$.
  The constructed Markov chain is always finite. 
  Therefore, Corollary~\ref{cor:DMR_finite} can be used to calculate the deadline miss rate.
\end{proof}

With the above discussions, the following Corollary concludes this section.
\begin{cor}
  \label{corollary:DMR-SupplyBoundFunctions-tarjan}
  If $X_{\bullet}$ generated by using our algorithm for the given
  supply functions is strongly connected (e.g., verified
  by adopting Tarjan's strong connected components
  algorithm~\cite{DBLP:journals/siamcomp/Tarjan72,DBLP:journals/ipl/NuutilaS94}),
  then $X_{\bullet}$ and the deadline miss rate of the periodic
  soft real-time task $\tau$ is \emph{upper bounded} by
  Eq.~(\ref{eq:DRM-pi}).
\end{cor}

\section{DMR under Supply Bound Functions}
\label{sec:DMR-supplybound}

This section sketches the key 
steps to construct the corresponding Markov chain 
that captures the execution of a soft
real-time task $\tau$ served by a GPC with supply bound functions
\mbox{$(\beta_1^u(t), \beta_1^l(t)), (\beta_2^u(t), \beta_2^l(t)), \ldots,
(\beta_j^u(t), \beta_j^l(t))$ $\forall j \in \Nbb$.} These steps are 
similar to those in Section~\ref{sec:DMR-supply-function} with the key
difference that the $service(j, t)$ used in
Section~\ref{sec:DMR-supply-function} has to be replaced by the upper
accumulative service $ service^u(j, t)$ and the lower accumulative
service $service^l(j, t)$ accordingly:

\begin{align*}
  service^u(j, t) = &\sum_{\ell=j}^{j+\floor{\frac{t}{T}}-1} \beta_\ell^u(T) +
  \beta_{j+\floor{\frac{t}{T}}}^u\left(t-T\floor{\frac{t}{T}}\right)\\
  service^l(j, t) = &\sum_{\ell=j}^{j+\floor{\frac{t}{T}}-1} \beta_\ell^l(T) +
  \beta_{j+\floor{\frac{t}{T}}}^l\left(t-T\floor{\frac{t}{T}}\right)
\end{align*}

We must 
replace $service(j+1, t)$ used in
Section~\ref{sec:DMR-supply-function} with $service^l(j+1, t)$ to
classify whether the job under consideration has a deadline miss or
not. When calculating the remaining execution time, we should keep as
much as 
the GPC permits.  Therefore, the remaining
execution time is calculated based on $service^u$.

We only explain the \emph{state expansions}, as the initialization
step is simpler and almost identical.  If $J_{j+1}$ has no deadline
miss, Eq.~(\ref{eq:remaining-hit-j-supply-function}) is replaced with
\begin{equation}
  \label{eq:remaining-hit-j-supply-bound-function}
  \max\{e_k+rem(s_{j,\ell}) - \beta_{j+1}^u(T), 0\}.
\end{equation}
If $J_{j+1}$ has a deadline miss when using $service^l(j+1, t)$ for its
estimation, Eq.~(\ref{eq:remaining-miss-j-supply-function}) is replaced
with
\begin{equation}
  \label{eq:remaining-miss-j-supply-bound-function}    
  \max\left\{    
    \min\left\{
      \begin{array}{l}
        rem(s_{j,\ell})+e_k,\\
        service^u(j+1, D+\delta)   
      \end{array}
    \right\} - \beta_{j+1}^l(T)\}, 0\right\}.            
\end{equation}

The algorithm in Section~\ref{sec:DMR-supply-function} can be directly
applied with this minor modification. The time complexity remains the
same by defining $W=\max_{j=1,2,\ldots,Q} service^u(j, D+\delta)$.

\begin{exmpl}
  \label{example:Markov-chain-supply-bound-functions}
  Consider the soft real-time task $\tau$ in
  Example~\ref{example:DMR-N-variable} with $T=4$, $D=4$, $\delta=1$,
  $\Pbb(C_j=2)=p=0.5$ and \mbox{$\Pbb(C_j=3)=0.5$}.  Suppose that $\tau$ is
  served by a GPC with the upper and lower supply bound functions described in Example~\ref{example:supply-bound-functions-v2}.
  Figure~\ref{fig:markov-construction-example-supply-bound-functions}
  shows the resulting Markov chain. 
  The construction
  of the initial states and of
  $s_{1,1}, s_{1,2}, s_{2,1}, s_{2,2}, s_{3,2}$ is identical to 
  Example~\ref{example:missrate-ergodic-stationary}.  For the
  out-going edge of $s_{2,1}$, if $C_3$ is $2$, then it can be
  finished in time and transits to $s_{3,2}$. However, if $C_3$ is
  $3$, then $J_3$ misses its deadline. Furthermore, we know that
  $1+3 = service^u(3, D+\delta)$. Therefore, $s_{2,1}$ transits to
  $s_{3,1}$ with $50\%$ probability. For the out-going edge of
  $s_{3,1}$, if $C_4$ is $2$, then $J_4$ misses its deadline and since
  $1+2=service^u(4, D+\delta)=3$ it has one unit of remaining
  workload. That is, $s_{3,1}$ transits to $s_{1,1}$ with $50\%$
  probability if $C_4$ is $2$. If $C_4$ is $3$, then $J_4$ misses its
  deadline, and, since $1+3>service^u(4, D+\delta)=3$, it has one unit
  of remaining workload and one unit of remaining workload is
  dismissed. That is, $s_{3,1}$ transits to $s_{1,1}$ with $50\%$
  probability if $C_4$ is $3$. As a result, $s_{3,1}$ transits to
  $s_{1,1}$ with $100\%$ probability.

  This Markov chain is irreducible and finite, and the resulting DMR 
  is $1/3$.
  
  Again, if we fix $D=4$ and vary $\delta=0,1,2, \dots$, or if we fix $\delta=1$ and vary $D=4,5,6, \dots$, we can calculate the DMR. 
  The results are illustrated in Figure~\ref{fig:Markov-chain-supply-bound-functions}.
  We observe that the DMR when enlarging $\delta$ first increases and then seemingly converges to a fixed value.
  The DMR when enlarging $D$ decreases.
  However, it does not decrease monotonically but jumps up and down.
  This is counterintuitive as one may assume that with a larger deadline the DMR would always decrease. 
  The reason is simply that the computed DMR is an over-approximation which is more or less loose depending on the choice of $D$.
  \qed
\end{exmpl}

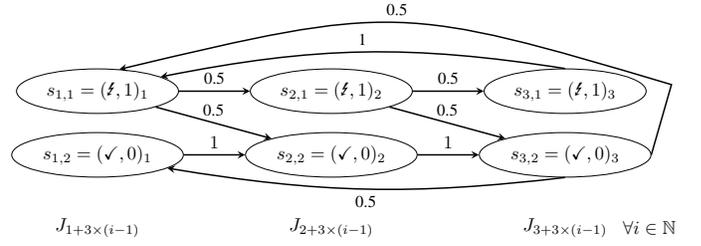
\begin{figure}[t]
  \centering
\scalebox{0.7}{

\begin{tikzpicture}[
  node distance = 0.35cm, auto, >=stealth,
  state/.style={ellipse, draw, minimum width=50pt},
  arrow/.style={->, thick}
]

\node[state] (node1A) {$s_{1,1} =(\Lightning, 1)_1$};
\node[state, below=of node1A] (node1B) {$s_{1,2}=(\checkmark, 0)_1$};
\node[state, right=of node1A, xshift=1cm] (node2A) {$s_{2,1}=(\Lightning, 1)_2$};
\node[state, below=of node2A] (node2B) {$s_{2,2}=(\checkmark, 0)_2$};
\node[state, right=of node2A, xshift=1cm] (node3A) {$s_{3,1}=(\Lightning, 1)_3$};
\node[state, below=of node3A] (node3B) {$s_{3,2}=(\checkmark, 0)_3$};

\node[below=of node1B, yshift=-.3cm] (nodeJ1) {$J_{1+3\times (i-1)}$};
\node[below=of node2B, yshift=-.3cm] (nodeJ2) {$J_{2+3\times (i-1)}$};
\node[below=of node3B, yshift=-.3cm] (nodeJ3) {$J_{3+3\times (i-1)}$};
\node[right=of nodeJ3, xshift=-.3cm] (nodeJ4) {$\forall i \in \Nbb$};

\node[right=of node3A] (node4A) {};
\node[right=of node3B] (node4B) {};

\draw[arrow] (node1A) to node[midway, above] {\small 0.5} (node2A);
\draw[arrow] (node1A) to node[midway, above] {\small 0.5} (node2B);

\draw[arrow] (node1B) to node[midway, above] {$1$} (node2B);

\draw[arrow] (node2A) to node[midway, above] {\small 0.5} (node3A);
\draw[arrow] (node2A) to node[midway, above] {\small 0.5} (node3B);

\draw[arrow] (node2B) to node[midway, above] {$1$} (node3B);

\draw[arrow] (node3A.north) to [bend right=10] node [midway, above]  {\small 1} (node1A);

\draw[arrow] (node3B.east) to  (node4A.north) to [bend right=15, looseness=1.3] node [midway, above]  {\small 0.5} (node1A.45);
\draw[arrow] (node3B.south) to [bend left=8] node [midway, below]  {\small 0.5} (node1B);
\end{tikzpicture}
}
\caption{Finite Markov chain of an infinite amount of jobs for
  Example~\ref{example:Markov-chain-supply-bound-functions}.}
  \label{fig:markov-construction-example-supply-bound-functions}
\end{figure}

\begin{figure}[t]
	\centering
	\subfloat[DMR by varying $\delta$]{
		\includegraphics[width=0.5\columnwidth]{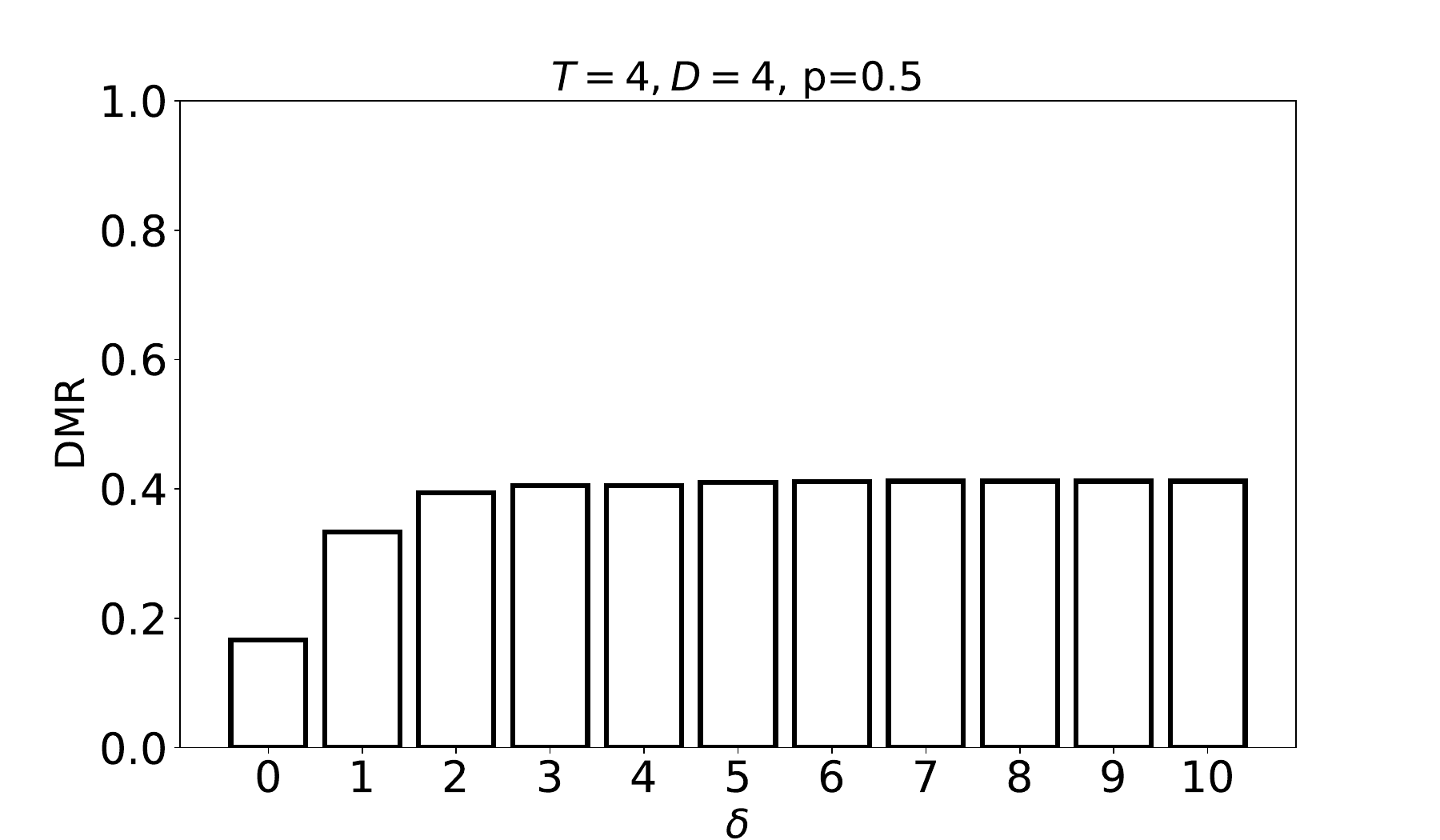}
	}
	\subfloat[DMR by varying $D$]{
		\includegraphics[width=0.5\columnwidth]{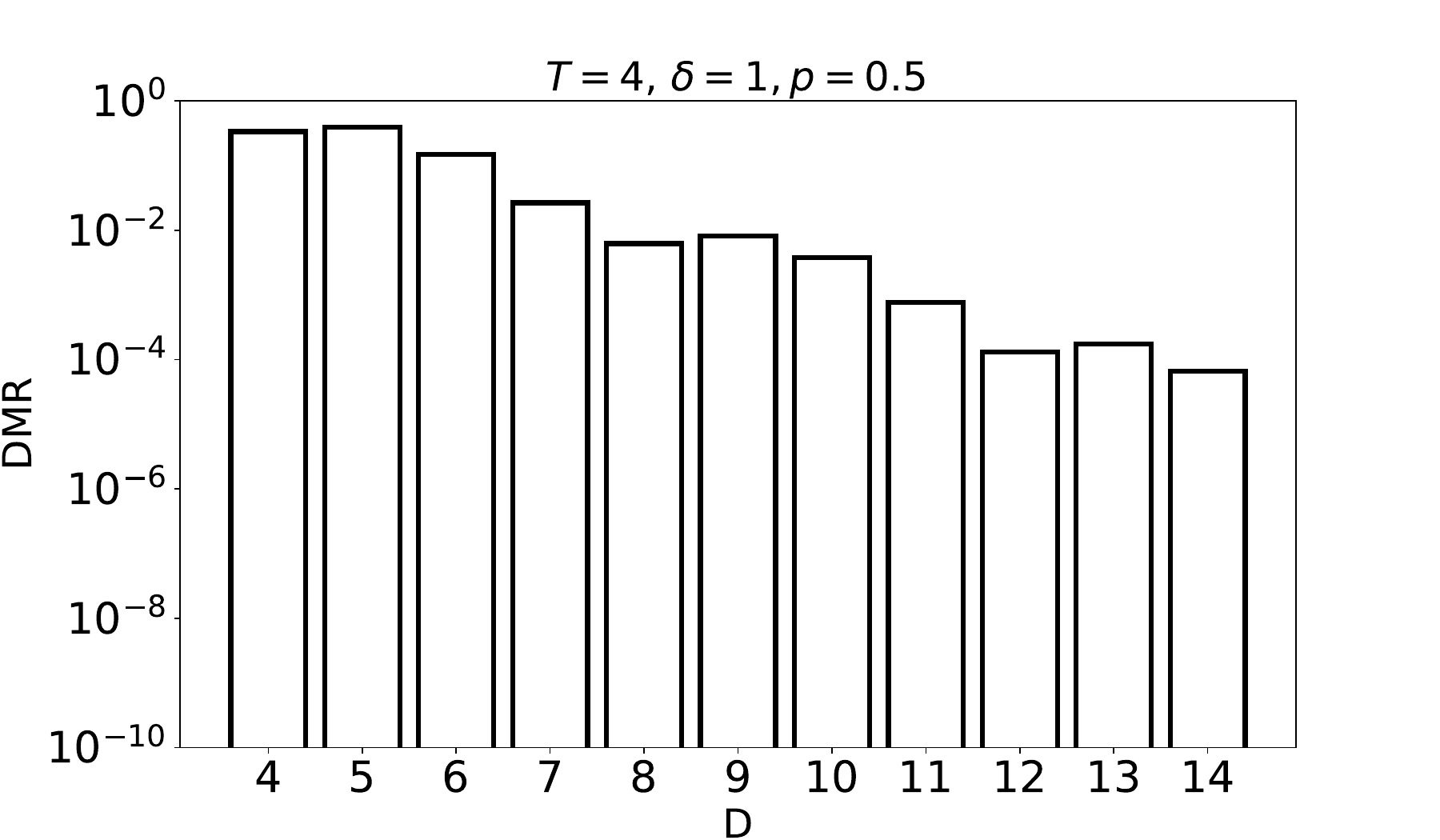}
	}
	\caption{DMR for Example~\ref{example:Markov-chain-supply-bound-functions}.}
	\label{fig:Markov-chain-supply-bound-functions}
\end{figure}

As in the case with supply functions, there is a canonical mapping from the realizations of $C_1, \dots, C_N$ to the realizations (states) of $X_1, \dots, X_N$ for the case with supply bound functions as well, following the construction of the Markov chain.
Although this mapping does not necessarily preserve the deadline miss rate, 
it always maps a realization of $C_1, \dots, C_N$ to a realization of $X_1, \dots, X_N$ with the same or higher deadline miss rate, which we prove in the following lemma.

\begin{lem}
  \label{lemma:safe-over-approximation-supply-bound}
  Let $N\in \Nbb$, let $(s_{1\Mod{Q},\ell_1}, \dots, s_{N\Mod{Q},\ell_N})$ be a realization of $X_1, \dots, X_N$ and let $\Psi$ be the number of states that indicate deadline misses in $(s_{1\Mod{Q},\ell_1}, \dots, s_{N\Mod{Q},\ell_N})$.
  All realization of $C_1, \dots, C_N$ that lead to the realization $(s_{1\Mod{Q},\ell_1}, \dots, s_{N\Mod{Q},\ell_N})$ result in a job sequence of at most $\Psi$ deadline misses.
\end{lem}


\begin{proof}
  We prove that this lemma holds for all $N\in \Nbb$ by induction over $N$.

  \textbf{Base state ($N=1$):} 
  Let $s_{1\Mod{Q},\ell_1}$ be a realization of $X_1$. 
  By construction, all realization of $C_1$ lead to jobs that have at most the backlog stated in $s_{1\Mod{Q},\ell_1}$. 
  The jobs can only have deadline misses if $s_{1\Mod{Q},\ell_1}$ indicates a deadline miss.

  \textbf{Induction step ($N-1 \mapsto N$):}
  Let $(s_{1\Mod{Q},\ell_1}, \dots, s_{N\Mod{Q},\ell_N})$ be a realization of $X_1, \dots, X_N$. 
  Then $(s_{1\Mod{Q},\ell_1}, \dots, s_{N-1\Mod{Q},\ell_{N-1}})$ is a realization of $X_1, \dots, X_{N-1}$.
  Let $\Rcal$ be the set of all realizations of $C_1, \dots, C_N$ that lead to $(s_{1\Mod{Q},\ell_1}, \dots, s_{N\Mod{Q},\ell_{N}})$.
  The first $N-1$ entries of any realization in $\Rcal$ lead to $(s_{1\Mod{Q},\ell_1}, \dots, s_{N-1\Mod{Q},\ell_{N-1}})$.
  Therefore, the first $N-1$ jobs of any job sequence obtained by $r \in \Rcal$ have at most the same number of deadline misses as $(s_{1\Mod{Q},\ell_1}, \dots, s_{N-1\Mod{Q},\ell_{N-1}})$ by induction.
  If any of the realizations of $C_N$ leads to a job sequence of an additional deadline miss, then 
  $s_{N\Mod{Q},\ell_{N}}$ indicates a deadline miss as well, by construction of the Markov chain.
  This proves that the number of deadline misses in $\Rcal$ is at most the same as the number of deadline misses in $(s_{1\Mod{Q},\ell_1}, \dots, s_{N\Mod{Q},\ell_{N}})$.
\end{proof}


\begin{thm}
  \label{theorem:DMR-SupplyBoundFunctions}
  If $X_{\bullet}$ generated by using our algorithm for the given
  probabilistic supply functions is irreducible, then the
  deadline miss rate of the periodic soft real-time task $\tau$ is
  \emph{upper bounded} by Eq.~(\ref{eq:DRM-pi}).
\end{thm}

\begin{proof}
  By Lemma~\ref{lemma:safe-over-approximation-supply-bound} the number
  of deadline misses in $X_{\bullet}$ is 
  not smaller than
  the number of deadline misses of task $\tau$.
  Since the constructed Markov chain is always finite, 
  Corollary~\ref{cor:DMR_finite} can be used to calculate the deadline miss rate of $X_{\bullet}$.
\end{proof}

Similar to Corollary~\ref{corollary:DMR-SupplyBoundFunctions-tarjan},
whether $X_{\bullet}$ is strongly connected or not can be verified by
Tarjan's strong connect components
algorithm~\cite{DBLP:journals/siamcomp/Tarjan72,DBLP:journals/ipl/NuutilaS94}.

\section{Remarks}
\label{sec:remarks}

The adoption of GPC in this paper allows us to analyze the deadline
miss rate of a soft real-time task for several scenarios, including
preemptive fixed-priority scheduling, reservation servers, and TDMA.
In order to simplify the presentation of this paper, we assume that the
supply functions and supply bound
functions 
are all specified in a segmented manner (with a segment
length of $T$) and are repeated every $Q$ segments. We note that this
simplification is 
 not necessary. In fact, 
when considering supply functions, we only need to identify the accumulative services
in time intervals
$[(j-1)T, jT),$ $[(j-1)T, (j-1)T+D), [(j-1)T, (j-1)T+D+\delta)$ for all
$j\in \Nbb$. This can be done in the deterministic manner (like in
Section~\ref{sec:DMR-supply-function}) 
or approximated manner (like in
Section~\ref{sec:DMR-supplybound}).  As long as the construction
process of the Markov chain can clearly identify the repetitive
pattern and the precise information of the supply functions and
supply bound functions, 
then the
algorithms to construct the Markov chains in
Sections~\ref{sec:DMR-supply-function} 
and~\ref{sec:DMR-supplybound}
can be applied (with some minor modifications). 

The fundamental approach presented in this work is easily extendable to more complex scenarios, if a higher number of states can be tolerated. 
For example, one may consider \emph{probabilistic supply functions}, where supply 
functions $\beta_j$ are randomly drawn from $\beta_j^1, \dots, \beta_j^g$ according to a discrete probability distribution.
If $\beta_1, \beta_2, \dots$ are i.i.d. and $D\leq D+\delta\leq T$, we only need to put
$(\{\checkmark, \Lightning\},w , \gamma)_j$ for
$\gamma \in \setof{1,2,\ldots,g}$ to further indicate that this state
uses the supply function $\beta_j^{\gamma}$. For $D+\delta > T$, we need to further consider the realization of $\ceiling{\frac{D+\delta}{T}}$ i.i.d. supply functions.
The Markov chain construction can be easily extended by not only considering the $h$ different
execution times of the task $\tau$ but also the $g$ different supply
functions.
The deadline miss rate can then be computed in a similar manner by calculating the stationary distribution.

We utilize several examples to demonstrate how the Markov chain can be
constructed. Specifically, the results in
Figure~\ref{fig:missrate-different-prob-example} and
Figure~\ref{fig:DMR-example-varying-delta-And-D} show that the DMR of
the soft real-time task gets higher if its expected execution time is
higher (by keeping the same $h$ different execution times), and that the
DMR gets higher if the dismiss point is longer or the relative
deadline is shorter. To the best of our knowledge, \emph{this is the
  first result regarding the tradeoff of the DMR and the dismiss
  point}. Most results in the literature for DMR analysis consider
that the soft real-time jobs continue to execute even after their
deadline misses. However, this may have a negative impact on the DMR
as shown in Figure~\ref{fig:DMR-example-varying-delta-And-D}.  In some
of the evaluated cases (not shown in this paper due to space constraints), 
a high $\delta$ 
resulted in a very high DMR whilst a low $\delta$ 
could keep a low DMR.  We note that the close-forms of Examples~\ref{example:DMR-N-variable}~and~\ref{example:missrate-ergodic-stationary-different-probability} can be derived by calculating the stationary distribution using Eq.~(\ref{eq:stationary-distribution}) under the law of large numbers.
These examples only serve for illustrating the concepts.

\begin{figure}[t]
	\includegraphics[width=0.95\columnwidth]{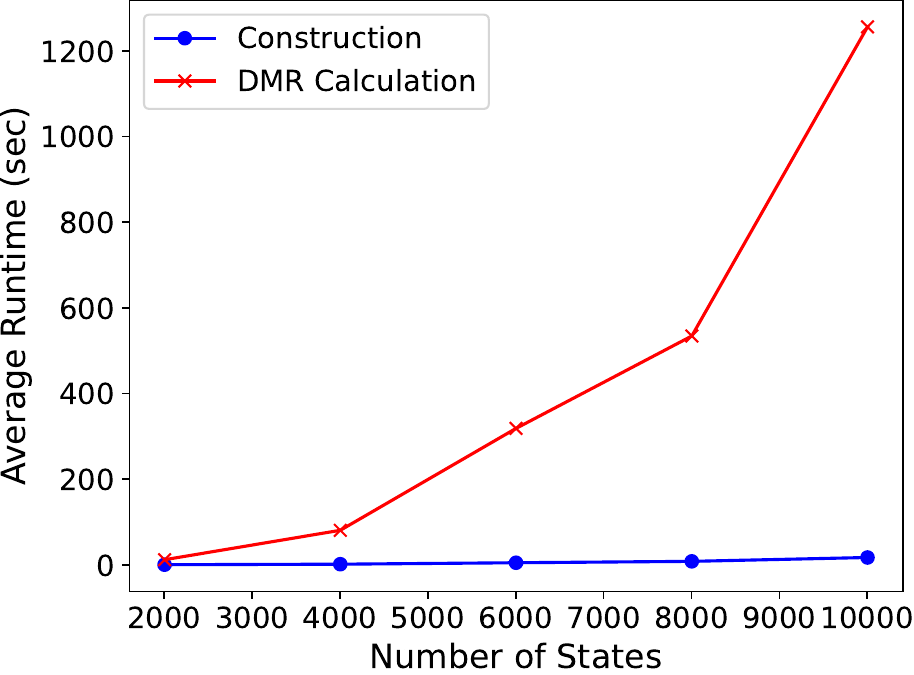}
	\caption{Average runtime over number of states}
	\label{fig:time}
\end{figure}

The scalability of our approach 
depends on the number of states in the constructed Markov chain, and the number of states are determined by the interplay between the factors in the considered system. We select $\delta$ as the knob for the setup as Example~\ref{example:DMR-N-variable}. The cardinality of the knob ranged from $1000$ to $5000$ in steps of $1000$, which leads to the number of steps in $\{2003,4003,6003,8003, 10003\}$. With respect to the cardinality, Figure~\ref{fig:time} shows the average runtime for the construction of the Markov chain and the calculation of DMR via Algorithm~\ref{alg:procedure}. The implementation of Algorithm~\ref{alg:procedure} utilizes the \texttt{scipy.sparse.linalg.eigs} library to numerically find eigenvectors of a sparse matrix in Python. All the tests were conducted on a laptop with an Intel i7-10610U. With the same delta, we also have tested other configurations with larger $T$, $D$, and $h$, for instance, $T=100$ and $h=10$. They all resulted in a smaller number of states. Since these results do not provide further insights, they are not reported here.


One feature that is probably not obvious in the model is that we do not
need $e_h$ (i.e., the worst-case execution time) of the soft real-time
task $\tau$ to be bounded, as long as the dismiss point $\delta$ is
finite and $\sum_{k=1}^{h-1} \mathbb{P}(C_j = e_k)$ can be
specified. Whenever we have to evaluate the realization $e_h$ (with
unbounded execution time) for a job, it is considered as a deadline
miss, and the GPC is fully exhausted until its dismiss point. We do
not exploit this feature, but it can be potentially combined with the
research line of probabilistic
WCET~\cite{DBLP:journals/lites/DavisC19} and measurement-based
execution time approaches.

\section{Conclusion}
\label{sec:conclusion}

We consider an arbitrary-deadline periodic soft real-time task in a
uniprocessor system. After a job misses its deadline, it can still be
executed until a dismiss point.  The analysis of the deadline miss
rate in the long run is achieved by modeling the execution behavior of
the task as a Markov chain, whilst its convergence is supported by the
ergodic theory.  To the best of our knowledge, this is the first work
that allows to specify a task-specific dismiss point when analyzing
the deadline miss rate and beyond constant bandwidth servers (CBSes)
and non-preemptive schedules.  Our results open an interesting
research direction to explore the tradeoff of the usefulness of the
soft real-time task after its deadline miss and its DMR.

We limit our attention to the fundamental properties of the DMR of a
task in this paper.  More sophisticated deadline-miss-related
behaviors such as the probability of $k$ successive deadline misses
(as also pursued by Chen et al.~\cite{DBLP:conf/rtcsa/ChenBC18}) are
interesting future work.

\label{last-page}
\bibliographystyle{abbrv}
\bibliography{real-time}

\clearpage

\end{document}